\documentclass[10pt,journal]{IEEEtran}
\usepackage{dsfont,amsthm}
\usepackage{cite}
\usepackage{amsmath,amstext,amsfonts,amssymb}
\usepackage{enumerate}
\usepackage{mathtools}

\usepackage{amscd}
\usepackage{bm}
\usepackage{color}
\usepackage[final]{pdfpages}
\usepackage{multicol}
\usepackage{cite}

\newtheorem{theorem}{Theorem}
\newtheorem{lemma}[theorem]{Lemma}

\newtheorem{definition}{Definition}
\newtheorem{proposition}[theorem]{Proposition}

\newtheorem{corollary}[theorem]{Corollary}

\DeclareMathOperator*{\argmax}{arg\,max}
\renewcommand{\Pr}[1]{{\rm Pr}\left[#1\right]}
\newcommand{\E}[1]{{\rm E}\left[#1\right]}
\newcommand{\Var}[1]{{\rm Var}\left[#1\right]}
\newcommand{\set}[1]{\mathcal{#1}}

\newcommand{\cv}{{\bf c}}

\newcommand{\ev}{{\bf e}}

\newcommand{\mv}{{\bf m}}

\newcommand{\tv}{{\bf t}}

\newcommand{\vv}{{\bf v}}
\newcommand{\xv}{{\bf x}}
\newcommand{\yv}{{\bf y}}

\newcommand{\zerov}{{\bf 0}}

\newcommand{\Gm}{{\bf G}}

\newcommand{\Id}{{\bf I}}

\newcommand{\Qm}{{\bf Q}}

\newcommand{\Vm}{{\bf V}}
\newcommand{\Xm}{{\bf X}}
\newcommand{\Ym}{{\bf Y}}
\newcommand{\Zm}{{\bf Z}}

\newcommand{\Nc}{{\cal N}}

\newcommand{\Pc}{{\cal P}}

\renewcommand{\Re}{{\rm Re}}

\begin{document}

\title{Stability of Bernstein's Theorem and \\ Soft Doubling for Vector Gaussian Channels}

\author{
\IEEEauthorblockN{Mohammad Mahdi Mahvari,~\IEEEmembership{Student Member,~IEEE} and Gerhard Kramer,~\IEEEmembership{Fellow,~IEEE}}

\thanks{
Date of current version \today.
This work was supported by the 6G Future Lab Bavaria funded by the Bavarian State Ministry of Science and the Arts, the project 6G-life funded by the Germany Federal Ministry for Education and Research (BMBF), and by the German Research Foundation (DFG) through projects 421682817 and 509917421. This paper was presented in part at the 2023 IEEE Information Theory Workshop [DOI: 10.1109/ITW55543.2023.10161689].
}

\thanks{
The authors are with the Institute for Communications Engineering, School of Computation, Information and Technology, Technical University of Munich (TUM), 80333 Munich, Germany (e-mail: mahdi.mahvari@tum.de; gerhard.kramer@tum.de).
}
}

\maketitle
\thispagestyle{plain}
\pagestyle{plain}
\begin{abstract}
The stability of Bernstein's characterization of Gaussian distributions is extended to vectors by utilizing characteristic functions. Stability is used to develop a soft doubling argument that establishes the optimality of Gaussian vectors for certain communications channels with additive Gaussian noise, including two-receiver broadcast channels. One novelty is that the argument does not require the existence of distributions that achieve capacity.
\end{abstract}


\section{Introduction}
\label{sec:introduction}

The following characterization of vector Gaussian distributions builds on the work of Kac~\cite{Kac1939} and Bernstein~\cite{bernstein41} and is a particular case of the main results in~\cite{darmois53,Skitovic53,ghurye1962characterization}.

\begin{theorem}
\label{thm:Bernstein}
Consider the independent\footnote{
By ``independent," we mean statistical independence instead of, e.g., linear independence.} $d$-dimensional random vectors $\Xm_1$ and $\Xm_2$. If $\Xm_1+\Xm_2$ and $\Xm_1-\Xm_2$ are also independent, then $\Xm_1$ and $\Xm_2$ are Gaussian and have the same covariance matrix.
\end{theorem}

We refer to Theorem~\ref{thm:Bernstein} as Bernstein's theorem. The result has been used to establish the optimality of Gaussian functions and Gaussian random vectors for several inequalities, including inequalities with applications to reliable communications over channels with additive Gaussian noise (AGN). The following sections review applications of Bernstein's theorem and motivate stability theorems.

\subsection{Applications of Bernstein's Theorem}
\label{subsec:app-Bernstein}

Lieb~\cite{lieb1990gaussian} used the {``}$O(2)$ rotation invariance of products of centered Gaussians'' to show that Gaussian functions achieve equality in the generalized Brascamp-Lieb inequality \cite{brascamp1976best}. Lieb's method is closely related to Bernstein's theorem as it considers products of a function with the vector arguments $(\xv_1+\xv_2)/\sqrt{2}$ and $(\xv_1-\xv_2)/\sqrt{2}$. A particular case of the generalized Brascamp-Lieb inequality is Young's inequality which is met with equality by Gaussian functions. Carlen~\cite{carlen1991superadditivity} used Lieb's technique to prove that Gaussian functions achieve equality in the logarithmic Sobolev inequalities; he refers to the technique as a ``doubling trick''.

More recently, two doubling tricks were used to prove inequalities related to communications problems~\cite{nair2014extremal}:
one based on Bernstein's theorem and another on the central limit theorem (CLT); see also~\cite{courtade2014extremal,courtade2017strong,liu2017information,liu2018information,liu2018forward, anantharam2022unifying,aras2022entropy}. These doubling tricks help to characterize the capacity region, or capacity points, of vector Gaussian broadcast channels \cite{Geng-Nair-IT14,chong2016capacity,ramachandran2017feedback,goldfeld2019mimo,lau2022uniqueness}, multiaccess channels with feedback~\cite{Sula20}, relay channels \cite{el2022strengthened}, Z-interference channels \cite{gohari2021information, costa2020structure}, Gray-Wyner networks \cite{gastpar2019relaxed}, source coding problems \cite{bross2020source,xu2021vector}, and two-way wiretap channels \cite{wang2021secret}. 
The proof using Bernstein's theorem is stronger since it establishes the uniqueness of the capacity-achieving distribution. However, the CLT-based proof seems to apply more generally, e.g., the CLT was needed to determine the sum-rate capacity for multiaccess channels with feedback~\cite{Sula20}.

\subsection{Motivation}
\label{subsec:motivation}

A key step in applying Bernstein's theorem to communications problems is establishing the existence of distributions achieving rate tuples on the boundaries of capacity regions. The primary motivation for this paper was to investigate the necessity of this step because it has several restrictive traits. First, in practice, one can only approach rather than achieve the capacity of noisy channels, so requiring a capacity supremum to be a maximum seems more of mathematical than engineering relevance. Second, if a maximizing distribution is continuous, there is no guarantee that practical (moderate-size finite) modulation alphabets achieve rates close to capacity.

A third limitation is that the existence proof in~\cite{Geng-Nair-IT14} requires several technical theorems on the convergence of sequences of distributions, including Prokhorov's theorem~\cite{prokhorov56}, the converse~\cite{boos85} of the Scheff\'e-Riesz theorem~\cite{scheffe47}, L\'evy's continuity theorem, and a theorem of Godavarti-Hero~\cite{Godavarti04}. As a consequence, proving existence can be tedious, and most papers \cite{ramachandran2017feedback, el2022strengthened, gohari2021information, bross2020source, chong2016capacity, costa2020structure, gastpar2019relaxed, xu2021vector} simplify exposition by referring to~\cite{Geng-Nair-IT14} for the approach. Instead, we wish to have an accessible proof that requires only basic theory and considers individual distributions rather than sequences.

A second motivation for this paper was to extend the stability of Bernstein's theorem from scalars to vectors. For scalars, this stability is based on the stability of the Cauchy functional equation for bi-infinite~\cite{hyers1941stability} and finite~\cite{skof1983proprieta} intervals; treating vectors requires extensions to multivariate functions~\cite{kominek1989local}. We also use stability results for vector differential entropies based on individual distributions~\cite{Ghourchian2017}.

\subsection{Stability of Cram\'er's Theorem}
\label{subsec:stability}

Stability for statistical independence has a long history~\cite{lukacs77aap}. For example, a theorem of Cram\'er~\cite{cramer36math} states that the sum $X=X_1+X_2$ of independent $X_1$ and $X_2$ is Gaussian if and only if $X_1$ and $X_2$ are Gaussian. A corresponding stability theorem is due to Sapogov~\cite{sapogov1951izv,sapogov1955lenin}, see also~\cite[Sec.~3]{lukacs77aap} that cites~\cite[p.~100]{levy37}. To state his result, consider the Gaussian cumulative distribution function (c.d.f.) with zero mean and unit variance, namely
\begin{align}
    \mathcal{F}(x) = \frac{1}{\sqrt{2\pi}} \int_{-\infty}^x e^{-y^2/2} \, dy, \quad x\in\mathbb{R}.
\end{align}
Similarly, define the Gaussian c.d.f.s $\mathcal{F}_i$ with means $m_i$ and variances $\sigma_i^2$, $i=1,2$. Let $\E{X}$ and $\Var{X}=\E{(X-\E{X})^2}$ denote the expectation and variance of $X$.

Consider the uniform distance, or Kolmogorov distance~\cite[eq.~(2.1)]{lukacs77aap}, between between the c.d.f.s $F$ and $G$:
\begin{align}
    d(F,G) = \sup_{x\in\mathbb R} |F(x) - G(x)| .
    \label{eq:uniform-metric}
\end{align}
Now if $X=X_1+X_2$ is approximately Gaussian, in the sense that its c.d.f. $F_X$ satisfies
$d( F_X,\mathcal F) < \epsilon$
for some $\epsilon$ satisfying $0<\epsilon<1$, then for $i=1,2$ we have
\begin{align}
    d\left( F_{X_i}, \mathcal{F}_i \right) <  c \, \sigma_i^{-3/4} (-\ln \epsilon)^{-1/8}
    \label{eq:sapogov}
\end{align}
where $c$ is a positive constant independent of $\epsilon$, $m_i\approx \E{X_i}$ and $\sigma_i^2 \approx \Var{X_i}$; see~\cite[Eq. (3.1)]{lukacs77aap} for the precise definitions.  Sapogov later \cite{sapogov1959izv} improved the right-hand side of~\eqref{eq:sapogov} to scale as $(-\ln \epsilon)^{-1/2}$ rather than $(-\ln \epsilon)^{-1/8}$
with $\epsilon$, and this scaling is the best possible in general~\cite{Maloshevsii68tvp}.

\subsection{Stability of Bernstein's Theorem}
\label{subsec:stability-bernstein}

Turning to Bernstein's theorem, several stability results for scalars are described in~\cite[Sec.~4]{lukacs77aap} with different assumptions on the sums and differences of independent $X_1$ and $X_2$. For example, consider~\cite[Thm.~4.4]{lukacs77aap} that uses the uniform metric, i.e., $X_1+X_2$ and $X_1-X_2$ are said to be $\epsilon$-dependent\footnote{This property is called $\epsilon$-independent in~\cite[Sec.~2]{lukacs77aap} but it seems more natural to identify independence with $0$-dependence rather than $0$-independence.} if
\begin{align}
    d\left( F_{X_1+X_2,X_1-X_2}, F_{X_1+X_2}F_{X_1-X_2} \right) < \epsilon
    \label{eq:uniform-metric2}
\end{align}
where the supremum in \eqref{eq:uniform-metric2} is over both real arguments of the c.d.f.s. Now suppose $0<\epsilon<1$ and $\E{|X_i|^{2(1+\delta)}}<\infty$ for $i=1,2$ and some $\delta$ satisfying $0<\delta\le 1$. Then \cite[Thm.~4.4]{lukacs77aap} states that \eqref{eq:uniform-metric2} implies
\begin{align}
    d\left( F_{X_i}, \mathcal{F}_i \right) <  c \, (-\ln \epsilon)^{-1/2}
    \label{eq:hhn}
\end{align}
for $i=1,2$, where $c$ is independent of $\epsilon$, $m_i=\E{X_i}$, and $\sigma_i^2 = (\Var{X_1}+\Var{X_2})/2$; see the text following~\eqref{eq:sapogov}.

The discussion in~\cite{lukacs77aap} describes several other stability metrics, including the L\'evy metric~\cite{levy37} that measures the distance between c.d.f.s $F$ and $G$ as
\begin{align}
    d_L(F,G) = \inf \big\{ h \ge 0:
    & F(x-h)-h \le G(x) \nonumber \\
    & \le F(x+h)+h
    \text{ for all } x \big\}.
    \label{eq:Levy-metric}
\end{align}
We instead follow Klebanov-Yanushkyavichyus~\cite{Klebanov1985,Klebanov1986} (see also~\cite{Yanushkevichius98}) and consider the uniform metric in the characteristic function (c.f.) domain. Let $j=\sqrt{-1}$ and let
\begin{align}
    f_X(t) = \E{e^{j t X}}, \;
    f_{X_1,X_2}(t_1,t_2) = \E{e^{j t_1 X_1 + j t_2 X_2}}
\end{align}
be the c.f.s of $X$ and the pair $(X_1,X_2)$, respectively. For example, the c.f. of a Gaussian distribution with mean $m$ and variance $\sigma^2$ is
\begin{align}
    \Phi(t) = e^{j m t - \frac{1}{2} \sigma^2 t^2}, \quad t\in\mathbb{R}.
\end{align}
$X_1$ and $X_2$ are said to be $\epsilon$-dependent in the c.f. domain if
\begin{align}
    d\left( f_{X_1,X_2}, f_{X_1}f_{X_2} \right) \le \epsilon 
    \label{eq:uniform-metric3}
\end{align}
where the supremum in \eqref{eq:uniform-metric3} is over both real arguments of the c.f.s. The paper~\cite{Klebanov1986} develops the following stability theorem. Let $\set P_\epsilon$ be the class of $(X_1,X_2)$ for which $X_1$ and $X_2$ are independent and $X_1+X_2$ and $X_1-X_2$ are $\epsilon$-dependent in the c.f. domain. Then we have (see~\cite[Thm.~1]{Klebanov1986})
\begin{align}
    c_1 \epsilon \le \sup_{(X_1,X_2)\in \set P_\epsilon} \,
    \max_{i=1,2} \; d\left( f_{X_i}, \Phi_i \right) \le c_2 \epsilon
    \label{eq:uniform-metric4}
\end{align}
for Gaussian c.f.s $\Phi_i$, $i=1,2$,
where $c_1$ and $c_2$ are positive constants independent of $\epsilon$. The bounds~\eqref{eq:uniform-metric4} imply that the scaling proportional to $\epsilon$ is generally the best possible.

\subsection{Multivariate Stability}
\label{subsec:multivariate-stability}

Gabovi{\v c}~\cite{Gabovic1976} established stability for a vector form of the Darmois-Skitovi{\v c} theorem~\cite{darmois53,Skitovic53} that generalizes Bernstein's theorem. However, there are several differences to the models and metrics of Sec.~\ref{subsec:stability-bernstein}. Consider the $d$-dimensional random vectors $\Xm_1$ and $\Xm_2$.
\begin{itemize}
\item Gabovi{\v c}~\cite{Gabovic1976} defines $\epsilon$-dependence for a vector form of the L\'evy metric \eqref{eq:Levy-metric} rather than a vector form of \eqref{eq:uniform-metric3};
\item $\Xm_1$ and $\Xm_2$ are permitted to be $\epsilon$-dependent (in the L\'evy metric) and not only $\Xm_1+\Xm_2$ and $\Xm_1-\Xm_2$; we also treat this case in Sec.~\ref{subsec:discussion} below;
\item the joint distribution of $\Xm_1,\Xm_2$ is shown to be near-Gaussian, whereas for independent $\Xm_1$ and $\Xm_2$ one may show that $\Xm_1$ and $\Xm_2$ are individually near-Gaussian;
\item the random vectors must satisfy a special condition ``to prevent the `leakage' of a significant probabilistic mass to infinity''~\cite[p.~5]{Gabovic1976}; this restriction seems to prevent the theory from fully generalizing Bernstein's theorem;
\item the stability in the L\'evy metric converges slowly in $\epsilon$ and behaves as
\begin{align}
    \frac{\big[\log \log \log (1/\epsilon)\big]^{(d+3)/4}}
    {\big[\log \log (1/\epsilon)\big]^{1/8}} \; ;
\end{align}
\item there is no claim of an identical covariance matrix for $\Xm_1$ and $\Xm_2$.
\end{itemize}

Thus, there are several advantages of studying stability in the c.f. domain. First, we need not \emph{a-priori} exclude certain random vectors, i.e., the stability theory generalizes Theorem~\ref{thm:Bernstein}. Second, we prove a common covariance matrix for sufficiently small $\epsilon$. Third, convergence is proportional to $\epsilon$ which is the best possible scaling, see \eqref{eq:uniform-metric4} and
Theorems~\ref{thm:boundedstability} and~\ref{thm:unboundedstability} below. Finally, we can relate $\epsilon$-dependence to mutual information; see Lemma~\ref{lemma:dt-to-mutualinfo} below. It is unclear whether one can improve the approach in~\cite{Gabovic1976} to give such properties.

\subsection{Organization}
\label{subsec:organization}

This paper has two main parts. The first part deals with the stability of Bernstein's theorem for random vectors. Sec.~\ref{sec:prelim} develops notation and reviews properties of multivariate c.f.s (Lemmas~\ref{lemma:c.f.upperbound} and~\ref{lemma:dt-to-mutualinfo}). Sec.~\ref{sec:stability-lemmas} develops several stability results for c.f.s and p.d.f.s (Lemmas~\ref{lemma:hyperstability}-\ref{lemma:pointwise-pdf-to-L1}). Sec.~\ref{sec:stability-theorems} states and proves our main stability theorems (Theorems~\ref{thm:boundedstability}-\ref{thm:unboundedstability-general}).

The second part of the paper applies the stability theory to AGN channels. Sec.~\ref{sec:soft-doubling} develops ``soft'' versions (Proposition~\ref{prop:agn} to Theorem~\ref{theorem:LV-IT07-Thm8}) of the ``hard'' doubling arguments in~\cite{Geng-Nair-IT14} for point-to-point channels, product channels, and two-receiver broadcast channels. Sec.~\ref{sec:conclusions} concludes the paper. 

Appendixes~\ref{appendix:a}-\ref{appendix:quadratic} develop results on the stability of Cauchy's functional equation (Lemmas~\ref{lemma:Hyers-stability}-\ref{lemma:multivariate-biadditve}) including for multivariate biadditive functions. Appendices~\ref{appendix:klebanov}-\ref{appendix:lemma:pointwise-pdf-to-L1} prove Lemmas~\ref{lemma:klebanov1}, \ref{lemma:pointwise-cf-to-pdf}, and~\ref{lemma:pointwise-pdf-to-L1}, respectively. Appendix~\ref{appendix:robust} treats a metric that is more restrictive than $\epsilon$-dependence (Lemma~\ref{lemma:robustly-epsilon-dependent}).

\section{Preliminaries}
\label{sec:prelim}

\subsection{Basic Notation}
\label{subsec:notation}

The $p$-norm for $d$-dimensional vectors is written as
\begin{align}
 \|\xv\|_p=\left(\sum_{i=1}^d |x_i|^p\right)^{1/p}
\end{align}
and we write $\|\xv\|_\infty=\max_{1\le i\le d}\,|x_i|$. We usually consider the 1-norm that we write as $\|\xv\|=\|\xv\|_1$. We have the bounds
\begin{align}
    & \|\xv\|_2 \, \le \|\xv\| \, \le \sqrt{d}\, \|\xv\|_2 
    \label{eq:normbound1} \\
    & \|\xv\|_\infty \, \le \|\xv\| \, \le d\, \|\xv\|_\infty .
    \label{eq:normbound2}
\end{align}
For complex-valued functions on $\mathbb R^d$ we write
\begin{align}
 \|f\|_p=\left(\int_{\mathbb R^d} |f(\tv)|^p \, d\tv\right)^{1/p}.
\end{align}
The $\ell_p$ distance of $f$ from $g$ is $\|f-g\|_p$. The volume of a ball of radius $r$ in $d$ dimensions with respect to the $p$-norm is
\begin{align}
    V_{p,d}(r) = \int_{\|\tv\|_p\le r} d\tv = \frac{\left(2 \Gamma(1+1/p)\right)^d}{\Gamma(1+d/p)} \, r^d
\end{align}
where $\Gamma$ is the gamma function. For example, if $p=1$, then
\begin{align}
    V_{1,d}(r) = \frac{2^d}{d!} \, r^d 
    \;\;\Rightarrow\;\; V_{1,d}(r) \le 2 \, r^d .
    \label{eq:volume-1-d-ball}
\end{align}

For a square matrix $\Qm$, we write $\det\Qm$ for the determinant of $\Qm$, and $\Qm'\preceq\Qm$ if $\Qm-\Qm'$ is positive semi-definite. The $d\times d$ identity matrix is written as $\Id_d$. The vector with zero entries except for a 1 in entry $i$ is written as $\ev_i$.

We write sets with calligraphic letters such as $\set E$. Set complements and direct products are written as $\set{E}^c$ and $\set{E}_1 \times \set{E}_2$, respectively.

The distribution, c.d.f., mean, and covariance matrix of $\Xm$ are written as $P_{\Xm}$, $F_{\Xm}$, $\mv_{\Xm}=\E{\Xm}$, and
\begin{align}
    \Qm_{\Xm}=\E{(\Xm-\mv_{\Xm})(\Xm-\mv_{\Xm})^T}
\end{align}
respectively, where $\tv^T$ is the transpose of $\tv$. The distribution $P_{\Xm}$ is absolutely continuous (a.c.) with respect to the Lebesgue measure if and only if a p.d.f. exists that we write as $p_{\Xm}$.

The notation $h(p)$, $h(\Xm)$, $I(\Xm;\Ym)$, and $D(p||q)$ refers to the differential entropy of the p.d.f. $p$, the differential entropy of $\Xm$, the mutual information of $\Xm$ and $\Ym$, and the informational divergence of the p.d.f.s $p$ and $q$, respectively. We often discard subscripts on probability distributions and other functions for notational convenience.

\subsection{Multivariate Characteristic Functions}
\label{subsec:cf}

The characteristic function (c.f.) of the $d$-dimensional real-valued $\Xm$ evaluated at $\tv \in {\mathbb R}^d$ is
\begin{align}
    f_{\Xm}(\tv) = \E{e^{j \tv^T \Xm}} . 
\end{align}
If the p.d.f. $p_{\Xm}$ exists then  $(p_{\Xm},f_{\Xm})$ can be interpreted as a Fourier transform pair. We will also consider
\begin{align}
    g_{\Xm}(\tv)=\ln{f_{\Xm}(\tv)}
\end{align}
and $g_{\Xm}$ is sometimes called the second c.f. of $\Xm$. If these functions have derivatives of all orders, then one may use a multivariate version of Taylor's theorem to write $f_{\Xm}(\tv)$ as an expansion of the moments of $\Xm$, and one can write $g_{\Xm}(\tv)$ as an expansion of the cumulants of $\Xm$.

The c.f. of the pair $\Xm_1,\Xm_2$ evaluated at $\tv_1,\tv_2$ is
\begin{align}
    f_{\Xm_1,\Xm_2}(\tv_1,\tv_2)
    = \E{e^{j \tv_1^T \Xm_1 + j \tv_2^T \Xm_2}} 
\end{align}
and similarly $g_{\Xm_1,\Xm_2}(\tv_1,\tv_2)=\ln{f_{\Xm_1,\Xm_2}(\tv_1,\tv_2)}$. Note that $\Xm_1$ and $\Xm_2$ need not have the same dimension.

\subsection{Properties of Characteristic Functions}
\label{subsec:cf-properties}

Four basic properties of c.f.s are as follows; see~\cite[p.~55]{ushakov1999selected}: $f_{\Xm}({\bf 0})=1$; $\left|f_{\Xm}(\tv)\right| \le 1$; $f_{\Xm}(-\tv)=f_{\Xm}(\tv)^*$ where $x^*$ is the complex conjugate of $x$; $f_{\Xm}$ is uniformly continuous and therefore non-vanishing in a region around $\tv={\bf 0}$. We thus also have: $\left|g_{\Xm}({\bf 0})\right|=0$; $\Re\{g_{\Xm}(\tv)\} \le 0$; $g_{\Xm}(-\tv)=g_{\Xm}(\tv)^*$; $\Re\{g_{\Xm}(\tv)\}>-\infty$ for a region around $\tv={\bf 0}$.

Another property is the following upper bound on $|f_{\Xm}(\tv)|$ for a region around $\tv={\bf 0}$.

\begin{lemma}[See {\cite[p.~114, Theorem~2.7.1]{ushakov1999selected}}]
\label{lemma:c.f.upperbound}
    Let $f_{\Xm}$ be the c.f. of a non-degenerate distribution in $\mathbb R^d$, i.e., the distribution is not concentrated on a hyperplane of dimension smaller than $d$. Then there exist positive constants $c,T$ such that
    \begin{align}
        |f_{\Xm}(\tv)| \leq 1 - c\|\tv\|^2 \text{ for } \|\tv\| \leq T.
        \label{eq:ushakov-theorem}
    \end{align}
\end{lemma}
\begin{proof}
See \cite[Theorem~2.7.1]{ushakov1999selected}. Note that\cite{ushakov1999selected} uses the 2-norm. However, since $c,T$ are generic, the bound \eqref{eq:normbound1} permits using the 1-norm. 
\end{proof}

We next state two properties of pairs of random vectors and define two versions of $\epsilon$-dependence.  The first property is that $\Xm_1$ and $\Xm_2$ are statistically independent if and only if $f_{\Xm_1,\Xm_2}$ factors as $f_{\Xm_1} f_{\Xm_2}$. Second, the following lemma relates the c.f.s of pairs of random vectors and their mutual information.

\begin{lemma}\label{lemma:dt-to-mutualinfo}
Suppose $\Xm_1, \Xm_2$ have dimensions $d_1,d_2$ and joint p.d.f. $p_{\Xm_1,\Xm_2}$. Then for all $\tv_1\in\mathbb R^{d_1}$ and $\tv_2\in\mathbb R^{d_2}$ we have
\begin{align}
    & \left| f_{\Xm_1, \Xm_2}(\tv_1, \tv_2) - f_{\Xm_1}(\tv_1) f_{\Xm_2}(\tv_2) \right| \nonumber \\
    & \quad \le \left\| p_{\Xm_1,\Xm_2} - p_{\Xm_1}p_{\Xm_2} \right\| \nonumber \\
    & \quad \le \sqrt{2I(\Xm_1;\Xm_2)}
    \label{eq:006}
\end{align}
where the mutual information is measured in nats.
\end{lemma}
\begin{proof}
One may write
\begin{align*}
    & \left| f_{\Xm_1, \Xm_2}(\tv_1, \tv_2) - f_{\Xm_1}(\tv_1) f_{\Xm_2}(\tv_2) \right| \nonumber \\
    & = \left| \int_{\mathbb{R}^{d_1+d_2}} e^{j\tv_1^T\xv_1}e^{j\tv_2^T\xv_2} \big[ p(\xv_1,\xv_2) - p(\xv_1) p(\xv_2) \big] d\xv_1 d\xv_2  \right| \\
    & \le \int_{\mathbb{R}^{d_1+d_2}} \left| p(\xv_1,\xv_2) - p(\xv_1) p(\xv_2) \right| d\xv_1 d\xv_2 \\
    & \le \sqrt{2 D(p_{\Xm_1,\Xm_2}||p_{\Xm_1} p_{\Xm_2})}
\end{align*}
where the final step is Pinsker's inequality~\cite[p.~44]{csiszar2011information} and the informational divergence is measured in nats.
\end{proof}

\begin{definition} \label{def:eT-dependent}
Let $\epsilon$ and $T$ be non-negative constants. The random vectors $\Xm_1$ and $\Xm_2$ are \emph{$(\epsilon,T)$-dependent} if 
\begin{align}
    \sup_{\|\tv_1\| \le T, \|\tv_2\| \le T}
    |f_{\Xm_1,\Xm_2}(\tv_1, \tv_2)-f_{\Xm_1}(\tv_1)f_{\Xm_2}(\tv_2)| \leq \epsilon . \label{eq:eT-dependent}
\end{align}
Similarly, $\Xm_1$ and $\Xm_2$ are \emph{$\epsilon$-dependent} if they are $(\epsilon,T)$-dependent for all non-negative $T$.
\end{definition}

The $\epsilon$-dependence of Definition~\ref{def:eT-dependent} can be interpreted as $(\epsilon,\infty)$-dependence. Also, $\Xm_1$ and $\Xm_2$ are $0$-dependent (or $(0,\infty)$-dependent) if and only if they are independent.

\subsection{Gaussian Vectors}
\label{subsec:cfs-Gauss}

We write $\Xm \sim \Nc(\mv_{\Xm},\Qm_{\Xm})$ if $\Xm$ is Gaussian with mean $\mv_{\Xm}$ and covariance matrix $\Qm_{\Xm}$, i.e., the p.d.f of $\Xm$ is
\begin{align}
    \phi_{\Xm}(\xv) = \frac{1}{\det\left(2\pi \Qm_{\Xm}\right)^{1/2}} e^{-\frac{1}{2}(\xv-\mv_{\Xm})^T \Qm_{\Xm}^{-1} (\xv-\mv_{\Xm})}
    \label{eq:pdf-Gauss}
\end{align}
where we assumed that $\Qm_{\Xm}$ is invertible. More generally, the Gaussian c.f. is
\begin{align}
    \Phi_{\Xm}(\tv) = e^{ \tv^T \left(j \mv_{\Xm} -\frac{1}{2} \Qm_{\Xm} \, \tv \right)}
    \label{eq:cf-Gauss}
\end{align}
and we have $|\Phi_{\Xm}(\tv)|=1$ if and only if $\tv$ lies in the null space of $\Qm_{\Xm}$. Otherwise, $|\Phi_{\Xm}(c\cdot\tv)|$ strictly decreases from 1 to 0 as $c$ increases from $c=0$ to $c=\infty$. Furthermore, we have
\begin{align}
    \Phi_{\Xm}(2\tv) = \Phi_{\Xm}(\tv)^2 \left| \Phi_{\Xm}(\tv) \right|^2
    \label{eq:cf-Gauss2}
\end{align}
so that for the integer $k\ge0$ we have $|\Phi_{\Xm}(2^k\cdot\tv)|=|\Phi_{\Xm}(\tv)|^{4^k}$ which decreases rapidly with $k$ if $\tv\ne{\bf 0}$.

Finally, a common approach to smooth an $\Xm$, e.g., having a degenerate distribution or having Dirac-delta components, is to add a non-degenerate Gaussian $\Zm$ with small covariances. The distribution of $\Ym=\Xm+\Zm$ is then a.c.\ with respect to the Lebesgue measure, since
\begin{align}
    p_{\Ym}(\yv) = \int_{\mathbb R^d} \phi_{\Zm}(\yv - \xv) \, P_{\Xm}(d\xv)
    \label{eq:ac-density}
\end{align}
serves as a p.d.f. of $\Ym$.\footnote{Some authors prefer to write $dP_{\Xm}(\xv)$ or $dP_{\Xm}$ instead of $P_{\Xm}(d\xv)$ in \eqref{eq:ac-density}. Also, the vector $\Zm$ need not be Gaussian, but it should be non-degenerate and have a p.d.f.}

\section{Stability Lemmas}
\label{sec:stability-lemmas}

This section states several stability lemmas. The first is a local stability of Cauchy's functional equation that we use to prove Theorem~\ref{thm:boundedstability}; see~\eqref{eq:bern:026} and~\eqref{eq:bern:040} below. The second is a multivariate version of a theorem from~\cite{Klebanov1986} that we use to prove Theorem~\ref{thm:unboundedstability}; see~\eqref{eq:005-09} below.

\begin{lemma}[See {\cite[Theorem~1]{kominek1989local}}]
\label{lemma:hyperstability}
Let $g:\left[-T,T\right)^d \rightarrow \mathbb  C$, $T >0$, be a continuous\footnote{It suffices that the projections of $g$ onto each coordinate have at least one continuous point, see Lemma \ref{lemma:kominek-local}.} function satisfying 
\begin{align}
    |g(\xv+\yv)-g(\xv)-g(\yv)| \leq \theta \label{eq:bern:023}
\end{align}
for all $\xv, \yv \in \left[-T,T\right)^d$ such that $\xv+\yv \in \left[-T,T\right)^d$, $\theta >0$. Then there is a continuous and linear function $G: \mathbb R^d \rightarrow \mathbb  C$ such that
\begin{align}
    |g(\xv)-G(\xv)| \leq (4d-1)\theta, \quad \forall \; \xv \in \left[-T,T\right)^d.
    \label{eq:bern:024}
\end{align}
Moreover, if \eqref{eq:bern:023} is valid for all $T>0$ (or $T=\infty$) then $G$ is unique. 
\end{lemma}
\begin{proof}
    See Lemma~\ref{lemma:kominek-local} in Appendix \ref{appendix:a}.
\end{proof}

In the following, to simplify notation we write $f_i$ and $g_i$ for $f_{\Xm_i}$ and $g_{\Xm_i}$, respectively, and similarly for $\mv_i$ and $\Qm_i$. We generally consider $d$-dimensional vectors.

\begin{lemma}[See {\cite[Eq.~(7)]{Klebanov1986}}]\label{lemma:klebanov1}
Suppose $\Xm_1$ and $\Xm_2$ are independent and $\Xm_1+\Xm_2$ and $\Xm_1-\Xm_2$ are $(\epsilon,T)$-dependent. Then for $\|\tv\|\le T$ and $i=1,2$ we have
\begin{align}
    f_i(2\tv) = f_i(\tv)^2 |f_i(\tv)|^2+r^{(3)}_{\epsilon, i}(\tv)
    \label{eq:005-08}
\end{align}
where $\left|r^{(3)}_{\epsilon, i}(\tv)\right| \le 5\epsilon$.
\end{lemma}
\begin{proof}
For general (perhaps dependent) $\Xm_1,\Xm_2$ we have
\begin{align}
    f_{\Xm_1+\Xm_2,\Xm_1-\Xm_2}(\tv_1,\tv_2)
    & = f_{\Xm_1,\Xm_2}(\tv_1+\tv_2,\tv_1-\tv_2) \label{eq:004-a}\\
    f_{\Xm_1+\Xm_2}(\tv_1) f_{\Xm_1-\Xm_2}(\tv_2)
    & = f_{\Xm_1,\Xm_2}(\tv_1,\tv_1) f_{\Xm_1,\Xm_2}(\tv_2,-\tv_2)\label{eq:004-b}
\end{align}
and hence for independent $\Xm_1,\Xm_2$, according to \eqref{eq:eT-dependent} we have
\begin{align}
    &f_1(\tv_1+\tv_2) f_2(\tv_1-\tv_2)\nonumber\\
    &= f_1(\tv_1)f_1(\tv_2) f_2(\tv_1)f_2(-\tv_2)+r_{\epsilon}(\tv_1, \tv_2)
    \label{eq:bern:004}
\end{align}
where
\begin{align}
    |r_{\epsilon}(\tv_1, \tv_2)| \le \epsilon \text{ and }
    \|\tv_i\| \le T, \; i=1, 2.
    \label{eq:bern:004-r}
\end{align}
For the remaining steps of the proof, see Appendix \ref{appendix:klebanov}.
\end{proof}

We next state a result from \cite{Ghourchian2017} on the existence and continuity of differential entropy in the 1-norm. Given $\alpha,m,\nu> 0$, define
$(\alpha,\nu,m)-\mathcal{AC}^d$ to be the class of $d$-dimensional vectors whose distributions are a.c.\ with respect to the Lebesgue measure, and for which the corresponding p.d.f. $p_{\Ym}$ satisfies
\begin{align}
    \underset{\yv\in\mathbb R^d}{\text{ess sup}} \; p_{\Ym}(\yv) < m \;\text{ and }\;
    \E{\|\Ym\|^{\alpha}_{\alpha}} < \nu .
\end{align}

\begin{lemma}[See {\cite[Theorem~1]{Ghourchian2017}}]\label{lemma:stability-entropy}
Let $p$ and $q$ be the p.d.f.s of two random vectors in $(\alpha,\nu,m)-\mathcal{AC}^d$. Then the differential entropies $h(p)$ and $h(q)$ exist. Moreover, if the $\ell_1$ distance satisfies $\|p-q\|\le m$ then
\begin{align}
    |h(p) - h (q)|
    & \le \| p - q \| \cdot \big( c_1 - c_2 \log \| p - q\| \big)
\end{align}
where
\begin{align}
    c_1 & = \frac{d}{\alpha}\left|\log \frac{2\alpha\nu}{d}\right| + |\log (m e)| + \log \frac{e}{2}\nonumber \\
    & \quad + d \log\left[2\Gamma\left(1+\alpha^{-1}\right)\right] + \frac{d}{\alpha}+1 \\ c_2 & = \frac{d}{\alpha} + 2.
\end{align}
\end{lemma}

Finally, we develop two lemmas that convert a pointwise bound in the c.f. domain to bounds in the p.d.f. domain.

\begin{lemma}\label{lemma:pointwise-cf-to-pdf}
Consider $\Ym=\Xm+\Zm$ where $\Zm \sim \Nc(\zerov, \Qm_{\Zm})$ is non-degenerate and independent of $\Xm$. Let $\lambda_{\Zm,{\rm min}}$ be the smallest eigenvalue of $\Qm_{\Zm}$. Suppose $\epsilon<1-e^{-\lambda_{\Zm,{\rm min}}/2}$ and $\left|f_{\Ym}(\tv)-\Phi(\tv)\right|\le \epsilon$ for all $\tv \in {\mathbb R}^d$ and for some Gaussian c.f. $\Phi$ with Fourier transform $\phi$. Then for all $\yv\in\mathbb R^d$ we have
\begin{align}
    \left|p_{\Ym}(\yv)-\phi(\yv)\right| \le B_1(\epsilon) \label{eq:bern:052-03a}
\end{align}
where $B_1(\epsilon) \rightarrow 0$ as $\epsilon \rightarrow 0$ and $B_1(\epsilon)=0$ if $\epsilon=0$.
\end{lemma}
\begin{proof}
    See Appendix~\ref{appendix:c}.
\end{proof}

\begin{lemma}\label{lemma:pointwise-pdf-to-L1}
Consider $\Ym_p$ and $\Ym_q$ with finite second moments and respective p.d.f.s $p$ and $q$. Suppose we have
\begin{align}
    \left|p(\yv)-q(\yv)\right| \le B_1(\epsilon) \label{eq:bern:052-03b}
\end{align}    
for all $\yv\in\mathbb R^d$ where $B_1(\epsilon) \rightarrow 0$ as $\epsilon \rightarrow 0$ and $B_1(\epsilon)=0$ if $\epsilon=0$. Then we have
\begin{align}
    \left\|p-q\right\| \le B_2(\epsilon) \label{eq:bern:052-03c}
\end{align}
where $B_2(\epsilon) \rightarrow 0$ as $\epsilon \rightarrow 0$ and $B_2(\epsilon)=0$ if $\epsilon=0$. Moreover, if the fourth moments\footnote{The fourth moments arise because we applied the Cauchy-Schwarz inequality in step $(a)$ of \eqref{eq:appd-bound8b} below. Instead, using H\"older's inequality, one can weaken the requirement and permit the $2+\delta$ moments to be bounded for any $\delta>0$, see the final paragraph of Appendix~\ref{appendix:lemma:pointwise-pdf-to-L1}.} of $\Ym_q$ are also bounded then
\begin{align}
    \E{\Ym_q\Ym_q^T} \preceq \E{\Ym_p\Ym_p^T} + B_3(\epsilon) \, \Id_d
    \label{eq:bern:052-03d}
\end{align}
where $B_3(\epsilon) \rightarrow 0$ as $\epsilon \rightarrow 0$ and $B_3(\epsilon)=0$ if $\epsilon=0$.
\end{lemma}
\begin{proof}
    See Appendix~\ref{appendix:lemma:pointwise-pdf-to-L1}.
\end{proof}

\section{Stability Theorems}
\label{sec:stability-theorems}

This section proves stability theorems for $d$-dimensional random vectors. Theorem~\ref{thm:boundedstability} considers local stability for a finite interval around $\tv=\bf 0$. Theorem~\ref{thm:unboundedstability} extends Bernstein's Theorem to include local stability by generalizing the scalar theory in~\cite{Klebanov1986} to vectors. Theorem~\ref{thm:entropies} gives two stability results: one for differential entropy and one for correlation matrices. We emphasize that these theorems have a common covariance matrix $\widehat \Qm$, which is not the case in~\cite{Klebanov1986,Gabovic1976} and is important to develop further results for product channels in Sec.~\ref{sec:AGN-Product} and for broadcast channels in Sec.~\ref{sec:broadcast}. 

\begin{theorem}\label{thm:boundedstability}
    Suppose $\Xm_1$ and $\Xm_2$ are independent random vectors, and $\Xm_1+\Xm_2$ and $\Xm_1-\Xm_2$ are $(\epsilon,T)$-dependent.
    Also, suppose there is a constant $p>0$ such that
    \begin{align}
        |f_i(\tv)| \ge p \quad \text{for $\|\tv\| \le T$ and $i=1,2$.} \label{eq:bern:002}
    \end{align}
    Then for
    $0<\epsilon\le p^4/[360d^2(d+1)]$ and $\|\tv\| \le T/2$ we have
    \begin{align}
        \left|f_i(\tv)-\Phi_i(\tv)\right| \le C(\epsilon) \cdot |\Phi_i(\tv)|, \quad i=1,2
        \label{eq:bern:052}
    \end{align}
    where for some mean vectors $\widehat \mv_i$, $i=1,2$, and for some common covariance matrix $\widehat \Qm$ we have the Gaussian c.f.s
    \begin{align}
        \Phi_i(\tv) & = e^{ \tv^T \left(j \widehat \mv_i -\frac{1}{2} \widehat \Qm \, \tv \right)}, \quad i=1,2
        \label{eq:cf-Gauss-i}
    \end{align}
    and the error term is
    \begin{align}
        C(\epsilon) & = \frac{720 d^2(d+1)}{p^4} \, \epsilon.
        \label{eq:Cep}
    \end{align}
\end{theorem}
\begin{proof}
    See Sec.~\ref{subsec:stability1-proof}.
\end{proof}

\begin{theorem}\label{thm:unboundedstability}
    Suppose $\Xm_1$ and $\Xm_2$ are independent random vectors, and $\Xm_1+\Xm_2$ and $\Xm_1-\Xm_2$ are $\epsilon$-dependent.
    Then for all $\epsilon$ below some positive threshold, for all $\tv \in {\mathbb R}^d$, and for $i=1,2$ we have
    \begin{align}
        \left|f_i(\tv)-\Phi_i(\tv)\right| \le \tilde C \epsilon \label{eq:bern:052-03}
    \end{align}
    for the Gaussian c.f.s \eqref{eq:cf-Gauss-i}, and
    for a constant $\tilde C$ independent of $\epsilon$ and $\tv$. In particular, if $\epsilon=0$, then $\Xm_1$ and $\Xm_2$ are Gaussian with the same covariance matrix.
\end{theorem}
\begin{proof}
    See Sec.~\ref{subsec:stability2-proof}.
\end{proof}

\begin{theorem}\label{thm:entropies}
Consider the random vectors $\Ym_1=\Xm_1+\Zm_1$ and $\Ym_2=\Xm_2+\Zm_2$ where $\Xm_1,\Xm_2,\Zm_1,\Zm_2$ are mutually independent, $\Ym_1,\Ym_2$ have finite second moments, and the noise vectors $\Zm_i \sim \Nc(\zerov,\Qm_{\Zm_i})$, $i=1,2$, are non-degenerate.
Suppose $\Ym_1+\Ym_2$ and $\Ym_1-\Ym_2$ are $\epsilon$-dependent.
Then for all $\epsilon$ below some positive threshold and for $i=1,2$ we have
\begin{align}
    \left|h(\Ym_i)-h(\Ym_{g,i})\right| \le B(\epsilon) \label{eq:bern:052-03-entropy}
\end{align}
where the $\Ym_{g,i}$ are Gaussian with the same covariance matrix, and thus $h(\Ym_{g,1})=h(\Ym_{g,2})$, and
\begin{align}
    \E{\Ym_{g,i} \Ym_{g,i}^T}
    \preceq \E{\Ym_i \Ym_i^T} + B(\epsilon)\,\Id_d
    \label{eq:bern:052-03-entropy-Q}
\end{align}
where $B(\epsilon)\rightarrow 0$ as $\epsilon\rightarrow 0$ and $B(\epsilon)=0$ if $\epsilon=0$.
\end{theorem}
\begin{proof}
See Sec.~\ref{subsec:stability3-proof}.
\end{proof}

\subsection{Proof of Theorem~\ref{thm:boundedstability}}
\label{subsec:stability1-proof}

We modify the proof steps of~\cite{Klebanov1985,Klebanov1986} who attribute their approach to~\cite{linnik1960,sapogov1981stability,Klebanov1981}. Several steps require considerations particular to multivariate distributions, e.g., Lemma~\ref{lemma:c.f.upperbound} and properties of covariance matrices such as their null spaces, symmetry, and positive semi-definite ordering.

Consider again \eqref{eq:bern:004}-\eqref{eq:bern:004-r} and define the functions
\begin{align}
   g_3(\tv) & = -g_1(\tv)-g_2(\tv) \\
   g_4(\tv) & = -g_1(\tv)-g_2(-\tv).
\end{align}   
Taking logarithms in \eqref{eq:bern:004}, we have (see~\cite[Eq.~(8)]{Klebanov1985})
\begin{align}
    g_1(\tv_1+\tv_2)+g_2(\tv_1-\tv_2)+g_3(\tv_1)+g_4(\tv_2)=R_{\epsilon}(\tv_1,\tv_2)
    \label{eq:bern:005}
\end{align}
for $\|\tv_i\| \le T$, $i=1, 2$, where
\begin{align}
   R_{\epsilon}(\tv_1,\tv_2)=\ln{\left(1+\frac{r_{\epsilon}(\tv_1, \tv_2)}{f_1(\tv_1)f_1(\tv_2) f_2(\tv_1)f_2(-\tv_2)}\right)}.
    \label{eq:bern:006}
\end{align}
We may bound 
\begin{align}
    \left|R_{\epsilon}(\tv_1,\tv_2)\right| \le \frac{3}{2p^4} \epsilon \label{eq:bern:007}
\end{align}
which follows from $|z|/2\le |\ln(1+z)| \le 3|z|/2$ for complex $z$ with $|z| \le 1/2$ \cite[p.~165]{conway1978}.

Next, define
\begin{align}
    g^{\prime}_{i,\tv}(\xv)\coloneqq
    g_i(\xv+\tv)-g_i(\xv)-g_i(\tv),
    \quad i=1,2,3\label{eq:bern:010-01}
\end{align}
and observe that \eqref{eq:bern:005} is the same as
\begin{align}
    g^{\prime}_{1,\tv_1}(\tv_2)+g^{\prime}_{2,\tv_1}(-\tv_2)\!=\!R_{\epsilon}(\tv_1,\tv_2).
    \label{eq:bern:010-01a}
\end{align}
Also, substituting $\tv_1\leftarrow \tv_1+\tv$ in \eqref{eq:bern:005}, and subtracting \eqref{eq:bern:005} from the resulting expression, we have (see~\cite[Eq.~(10)]{Klebanov1985})
\begin{align}
    g_{1,\tv}^{\prime}(\tv_1+\tv_2)+g_{2,\tv}^{\prime}(\tv_1-\tv_2)+g_{3,\tv}^{\prime}(\tv_1)=R^{(1)}_{\epsilon}(\tv,\tv_1,\tv_2)
    \label{eq:bern:011}
\end{align}
for $\|\tv_1\| \le T$, $\|\tv_2\| \le T$, $\|\tv_1+\tv\| \le T$, where
\begin{align}
    |R^{(1)}_{\epsilon}(\tv,\tv_1,\tv_2)| \leq \frac{3}{p^4} \epsilon.
\end{align}
By replacing $\tv_1 \leftarrow \tv, \tv_2 \leftarrow \tv_2-\tv_1$ in \eqref{eq:bern:010-01a} and $\tv_2 \leftarrow \tv_1$ in \eqref{eq:bern:011} we have the respective
\begin{align}
    & g_{2,\tv}^{\prime}(\tv_1-\tv_2)=-g_{1,\tv}^{\prime}(\tv_2-\tv_1)+R_{\epsilon}(\tv,\tv_2-\tv_1) \label{eq:bern:013}\\
    & g_{3,\tv}^{\prime}(\tv_1)=-g_{1,\tv}^{\prime}(2\tv_1)+R^{(1)}_{\epsilon}(\tv,\tv_1,\tv_1)
    \label{eq:bern:014}
\end{align}
where the following inequalities should be satisfied:
\begin{align}
    & \|\tv\|\le T/2, \quad\|\tv_1\| \le T, \quad\|\tv_2\| \le T, \nonumber \\
    & \|\tv+\tv_1\| \le T, \quad\|\tv_2-\tv_1\| \le T.
    \label{eq:bern:016}
\end{align}
Now substitute \eqref{eq:bern:013} and \eqref{eq:bern:014} in \eqref{eq:bern:011} to obtain
\begin{align}
    &g_{1,\tv}^{\prime}(\tv_1+\tv_2)-g_{1,\tv}^{\prime}(\tv_2-\tv_1)+R_{\epsilon}(\tv,\tv_2-\tv_1) \nonumber \\
    & \quad -g_{1,\tv}^{\prime}(2\tv_1)
    +R^{(1)}_{\epsilon}(\tv,\tv_1,\tv_1)=R^{(1)}_{\epsilon}(\tv,\tv_1,\tv_2) .
    \label{eq:bern:017}
\end{align}
Then substituting $\xv:=\tv_2-\tv_1$, and $\yv:=2\tv_1$ gives (see~\cite[Eq.~(14)]{Klebanov1985})
\begin{align}
    &g_{1,\tv}^{\prime}(\xv+\yv)-g_{1,\tv}^{\prime}(\xv)-g_{1,\tv}^{\prime}(\yv) \nonumber \\
    & =R^{(1)}_{\epsilon}(\tv,\yv/2,\xv+\yv/2)-R_{\epsilon}(\tv,\xv)-R^{(1)}_{\epsilon}(\tv,\yv/2,\yv/2)\nonumber \\
    &\coloneqq R^{(2)}_{\epsilon}(\tv,\xv,\yv)
    \label{eq:bern:018}
\end{align}
where
\begin{align}
    & |R^{(2)}_{\epsilon}(\tv, \xv, \yv)| \leq \frac{15}{2p^4} \epsilon
    \label{eq:bern:019}
\end{align}
for
\begin{align}
    &\|\tv\|\leq T/2, \quad\|\xv\| \le T, \quad\|\yv\| \le 2T,\label{eq:bern:020-01}\\
    &\|\tv+\yv /2\| \leq T,  \quad\|\xv + \yv /2\| \le T.\label{eq:bern:020-02}
\end{align}

Now apply Lemma \ref{lemma:hyperstability} to \eqref{eq:bern:018} to obtain (see~\cite[Eq.~(22)]{Klebanov1985})
\begin{align}
    \left| g_{1,\tv}^{\prime}(\xv)-\cv_{\tv}^{\mathrm{T}}\xv \right| \le \frac{15(4d-1)}{2p^4} \epsilon
    \label{eq:bern:026}
\end{align}
for $\|\tv\| \le T/2, \|\xv\| \le T/2$ and where $\cv_{\tv} \in \mathbb  C^{d}$ depends on $\tv$. Moreover, the relation \eqref{eq:bern:010-01a} implies
\begin{align}
    \left| g_{2,\tv}^{\prime}(\xv)-\cv_{\tv}^{\mathrm{T}}\xv \right| \le \frac{15(4d-1)+3}{2p^4} \epsilon \le \frac{30d}{p^4} \epsilon .
    \label{eq:bern:026a}
\end{align}
Inserting \eqref{eq:bern:010-01} into \eqref{eq:bern:026} and \eqref{eq:bern:026a} gives
\begin{align}
    & g_{i}(\xv+\tv)-g_{i}(\xv)-g_{i}(\tv)-\cv_{\tv}^{\mathrm{T}}\xv = R^{(3)}_{\epsilon,i,1}(\tv, \xv)
    \label{eq:bern:027} \\
    & g_{i}(\xv+\tv)-g_{i}(\xv)-g_{i}(\tv)-\cv_{\xv}^{\mathrm{T}}\tv = R^{(3)}_{\epsilon,i,2}(\tv, \xv)
    \label{eq:bern:028}
\end{align}
for $i=1,2$ and $\|\tv\| \le T/2$, $\|\xv\| \le T/2$, where \eqref{eq:bern:028} follows because $g_{1,\tv}^{\prime}(\xv)$ is symmetric with respect to $\xv$ and $\tv$, and where for $i=1,2$ we have
\begin{align}
    \left| R^{(3)}_{\epsilon,i,1}(\tv, \xv) \right| \le \frac{30d}{p^4} \epsilon, 
    \quad \left| R^{(3)}_{\epsilon,i,2}(\tv, \xv) \right| \le \frac{30d}{p^4} \epsilon.
    \label{eq:bern:028-Rbounds}
\end{align}
We emphasize that the vectors $\cv_{\tv}$ and $\cv_{\xv}$ in \eqref{eq:bern:027}-\eqref{eq:bern:028} do not depend on $i$, which is important to establish a common covariance matrix for $\Xm_1$ and $\Xm_2$ in what follows.

We continue to work with $g_{1}(\tv)$ since the same steps follow for $g_{2}(\tv)$. Equations \eqref{eq:bern:027}-\eqref{eq:bern:028} suggest that the symmetric function
\begin{align}
    \tilde g_{1}(\tv,\xv): = g_{1,\tv}'(\xv) = g_{1}(\xv+\tv)-g_{1}(\xv)-g_{1}(\tv)
    \label{eq:bern:028a}
\end{align}
has an ``almost'' symmetric bilinear form, namely
\begin{align}
    \tilde g_{1}(\tv,\xv) = \sum_{k=1}^d \sum_{\ell=1}^d t_k \, x_\ell \, \underbrace{\tilde g_1(\ev_k,\ev_\ell)}_{\displaystyle := -\widetilde\Qm_{k,\ell}} = - \tv^T \widetilde\Qm \xv 
    \label{eq:bern:028b}
\end{align}
where $\widetilde\Qm^T=\widetilde\Qm$.
More precisely, note that $\|T\ev_k/2\|=T/2$ and we have (discarding the dependence of the errors on $\tv,\xv$ where convenient)
\begin{align}
    & \tilde g_1(\tv,\xv) \overset{(a)}{=} R_{\epsilon,1,2}^{(3)} + \cv_{\xv}^{\mathrm{T}} \left(\sum_{k=1}^d \frac{2 t_k}{T}\, \frac{T}{2} \ev_k\right) \nonumber \\
    & \overset{(b)}{=} R_{\epsilon,1,2}^{(3)} + \sum_{k=1}^d \frac{2 t_k}{T} \left[ \tilde R_{\epsilon,1,k}^{(3)} + \cv_{T\ev_k/2}^{\mathrm{T}} \left(\sum_{\ell=1}^d x_\ell \ev_\ell \right)  \right]
    \label{eq:bern:028c}
\end{align}
where step $(a)$ follows by \eqref{eq:bern:028} and step $(b)$ follows by using \eqref{eq:bern:027}-\eqref{eq:bern:028} to write
\begin{align}
    \cv_{\xv}^{\mathrm{T}} \left(\frac{T\ev_k}{2}\right) & = \cv_{T\ev_k/2}^{\mathrm{T}}\xv + \tilde R^{(3)}_{\epsilon,1,k}(\xv)
    \label{eq:bern:028d}
\end{align}
for $\tilde R^{(3)}_{\epsilon,1,k}(\xv):=R^{(3)}_{\epsilon,1,1}(T\ev_k/2,\xv)-R^{(3)}_{\epsilon,1,2}(T\ev_k/2,\xv)$. Observe that
\begin{align}
    \left|\tilde R^{(3)}_{\epsilon,1,k}(\xv)\right| \le \frac{60d}{p^4}\,\epsilon.
    \label{eq:bern:028e}
\end{align}
Repeating the same steps for $\tilde g_1(\xv,\tv)$ and averaging, we get the symmetric bilinear form\footnote{One may alternatively prove the symmetry and bilinearity by using the stability of quadratic functional equations, see Appendix \ref{appendix:quadratic}. In particular, one can apply Lemma~\ref{lemma:multivariate-biadditve} to \eqref{eq:bern:018} and obtain an expression like \eqref{eq:bern:028f} directly.} (see~\cite[Eq.~(24)]{Klebanov1985})
\begin{align}
    \tilde g_1(\tv,\xv) & = R^{(4)}_{\epsilon} + \sum_{k=1}^d \sum_{\ell=1}^d t_k \, x_\ell \, \underbrace{\frac{1}{T}\left(\cv_{T\ev_k/2}^{\mathrm{T}} \ev_\ell + \cv_{T\ev_\ell/2}^{\mathrm{T}} \ev_k \right)}_{\displaystyle := -\widetilde\Qm_{k,\ell}} \nonumber \\
    & = R^{(4)}_{\epsilon} - \tv^{\mathrm{T}}\widetilde\Qm\,\xv
    \label{eq:bern:028f}
\end{align}
where $\|\tv\|\le T/2$ gives $\left|\sum_{k=1}^d \frac{2 t_k}{T} \tilde R_{\epsilon,1,k}^{(3)}\right| \le \frac{60d}{p^4} \epsilon$ and thus
\begin{align}
    \left|R^{(4)}_{\epsilon}(\tv,\xv)\right| \le \frac{90d}{p^4}\,\epsilon.
    \label{eq:bern:028g}
\end{align}

Let $\widetilde\Qm_R,\widetilde\Qm_I$ be the real and imaginary parts of $\widetilde\Qm$. Note that $\tilde g_1(-\tv,-\xv)^*=\tilde g_1(\tv,\xv)$ so taking complex-conjugates in \eqref{eq:bern:028f} gives
\begin{align}
    -\tv^{\mathrm{T}}\widetilde\Qm\,\xv + R^{(4)}_{\epsilon}(\tv,\xv) = -(\tv^{\mathrm{T}}\widetilde\Qm \,\xv)^* + R^{(4)}_{\epsilon}(-\tv,-\xv)^*.
    \label{eq:bern:029a}
\end{align}
We thus have
\begin{align}
    -j \,\tv^T \widetilde\Qm_I \xv = \frac{1}{2} \left( R^{(4)}_{\epsilon}(-\tv, -\xv)^* - R^{(4)}_{\epsilon}(\tv, \xv) \right)
    \label{eq:bern:029b}
\end{align}
and therefore \eqref{eq:bern:028f} is
\begin{align}
    & g_{1}(\xv+\tv)-g_{1}(\xv)-g_{1}(\tv) \nonumber \\
    &= -\tv^{\mathrm{T}}\widetilde\Qm_R\xv + \frac{1}{2} \left( R^{(4)}_{\epsilon}(-\tv, -\xv)^* + R^{(4)}_{\epsilon}(\tv, \xv)\right) \label{eq:bern:029c}
\end{align}
where $\widetilde\Qm_R$ is symmetric. Now write the solution of \eqref{eq:bern:029c} in the form
\begin{align}
    g_1(\tv)=\tv^T \widetilde \mv_1 -\frac{1}{2}\tv^{\mathrm{T}}\widetilde\Qm_R\tv+g(\tv)
    \label{eq:bern:037}
\end{align}
where  $\widetilde \mv_1 \in \mathbb  C^d$ and $g(\tv)$ is some function.\footnote{The subscript of $\widetilde \mv_1$ emphasizes that this vector is specific to $g_1(\tv)$ while for $g_2(\tv)$ one may choose a vector $\widetilde \mv_2$ other than $\widetilde \mv_1$. On the other hand, the matrix $\widetilde \Qm_R$ is common to $g_1(\tv)$ and $g_2(\tv)$ since $\cv_{\tv}$ and $\cv_{\xv}$ are independent of $i$ in \eqref{eq:bern:027}-\eqref{eq:bern:028}.} Inserting \eqref{eq:bern:037} into \eqref{eq:bern:029c} gives
\begin{align}
    |g(\xv+\tv)-g(\xv)-g(\tv)| \le \frac{90d}{p^4}\,\epsilon.
    \label{eq:bern:039}
\end{align}
Again applying Lemma \ref{lemma:hyperstability} we obtain
\begin{align}
    g_1(\tv) & = \tv^T \widehat\mv_1 -\frac{1}{2}\tv^{\mathrm{T}}\widetilde\Qm_R\tv+R^{(5)}_{\epsilon}(\tv)
    \label{eq:bern:040}
\end{align}
where $\widehat \mv_1 \in \mathbb  C^d$ and 
\begin{align}
    \left|R^{(5)}_{\epsilon}(\tv)\right| \le \frac{90 d (4d-1)}{p^4}\,\epsilon.
    \label{eq:bern:041}
\end{align}
Let $\widehat\mv_{1,R},\widehat\mv_{1,I}$ be the real and imaginary parts of $\widehat \mv_1$. Since $g_1(-\tv)^*=g_1(\tv)$ we have
\begin{align}
    \tv^{\mathrm{T}}\widehat\mv_{1,R} = \frac{1}{2} \left( R^{(5)}_{\epsilon}(-\tv)^* - R^{(5)}_{\epsilon}(\tv) \right)
    \label{eq:bern:043}.
\end{align}
Combining \eqref{eq:bern:040} and \eqref{eq:bern:043}, we obtain (see~\cite[Eq.~(32)]{Klebanov1985})
\begin{align}
    g_1(\tv)=j\tv^{\mathrm{T}}\widehat\mv_{1,I}-\frac{1}{2}\tv^{\mathrm{T}}\widetilde\Qm_R\tv+R^{(6)}_{\epsilon}(\tv)
    \label{eq:bern:045}
\end{align}
where $\|\tv\| \le T/2$ and
\begin{align}
    \left|R^{(6)}_{\epsilon}(\tv)\right| \le \frac{360 d^2}{p^4}\,\epsilon .
    \label{eq:bern:046}
\end{align}

Next, although $\widetilde \Qm_R$ is symmetric, it might not be a covariance matrix, i.e., it might not be positive semi-definite. To address this issue, note that $|f_1(\tv)|\le 1$ implies $\Re\{g_1(\tv)\}\le 0$ and therefore \eqref{eq:bern:045}-\eqref{eq:bern:046} give
\begin{align}
    \tv^{\mathrm{T}}\widetilde\Qm_R\tv \ge 2\Re\{R^{(6)}_{\epsilon}(\tv)\} \ge - \frac{720 d^2}{p^4}\,\epsilon.
    \label{eq:bern:046-cov}
\end{align}
We use $d \|\tv\|_2^2/\|\tv\|^2\ge 1$ to bound
\begin{align}
    \min_{\|\tv\|=T/2} \tv^{\mathrm{T}} \underbrace{\left( \widetilde\Qm_R + \frac{d}{(T/2)^2} \frac{720 d^2}{p^4}\,\epsilon \, \Id_d \right)}_{\textstyle := \widehat \Qm_R} \tv \ge 0
    \label{eq:bern:047}
\end{align}
and thus have $\widehat \Qm_R\succeq{\bf 0}$. We now rewrite \eqref{eq:bern:045} as
\begin{align}
    g_1(\tv)=j\tv^{\mathrm{T}}\widehat\mv_{1,I}-\frac{1}{2}\tv^{\mathrm{T}}\widehat\Qm_R\tv+R^{(7)}_{\epsilon}(\tv)
    \label{eq:bern:048}
\end{align}
where we use $\|\tv\|_2 \le \|\tv\| \leq T/2$ to bound
\begin{align}
    \left|R^{(7)}_{\epsilon}(\tv)\right| \le \frac{360 d^2(d+1)}{p^4}\,\epsilon .
    \label{eq:bern:049}
\end{align}
We thus have
\begin{align}
    \frac{360 d^2(d+1)}{p^4}\,\epsilon &\ge \left|g_1(\tv) - \tv^{\mathrm{T}}\left( j\widehat\mv_{1,I}-\frac{1}{2}\widehat\Qm_R\tv\right) \right|
    \nonumber\\
    & \overset{(a)}{=} \left|\ln f_1(\tv) - \ln \Phi_1(\tv) \right|
    \nonumber\\
    & \overset{(b)}{\ge} \ln\left(1+\left|\frac{f_1(\tv)}{\Phi_1(\tv)}-1\right|\right)
    \label{eq:bern:051}
\end{align}
where $(a)$ follows by choosing $\Phi_1(\tv)$ to be the Gaussian c.f. \eqref{eq:cf-Gauss-i} with $i=1$, $\widehat m_1=\widehat m_{1,I}$ and $\widehat\Qm=\widehat\Qm_R$; $(b)$ follows since $|\ln z| \geq \ln(1+|z-1|)$ for any complex value $z$.\footnote{The series expansion of $\exp(z')$ gives $|\exp(z')-1|\le \exp|z'|-1$. Now choose $z'=\ln z$, rearrange, and take logarithms.} We obtain a similar result for $g_2(\tv)$ and thus have
\begin{align}
    \left|\frac{f_i(\tv)-\Phi_i(\tv)}{\Phi_i(\tv)}\right|
    \le \exp{\left(\frac{360 d^2(d+1)}{p^4}\,\epsilon \right)}-1
    \label{eq:bern:052-01}
\end{align}
for $i=1,2$, $\|\tv\| \leq T/2$,  and the Gaussian c.f.s \eqref{eq:cf-Gauss-i}.

To obtain \eqref{eq:bern:052}, observe that for $0\le x\le 1$ we have $1+x \le e^x\le 1+2x$. Thus, if
\begin{align}
    \epsilon\le \frac{p^4}{360 d^2(d+1)}
    \label{eq:epsi-bounds}
\end{align}
then we obtain $C(\epsilon)$ as in \eqref{eq:Cep}.

\subsection{Proof of Theorem~\ref{thm:unboundedstability}}
\label{subsec:stability2-proof}

Consider first $\epsilon=0$. If we follow the steps of Sec.~\ref{subsec:stability1-proof}, in \eqref{eq:bern:007} we may set $R_{\epsilon}(\tv_1,\tv_2)=0$ and can therefore follow all the remaining steps as if $p>0$ even if $|f_1(\tv)|=0$ for some $\tv$. In other words, we can discard the constraint \eqref{eq:bern:002} and set $T=\infty$ in Theorem~\ref{thm:boundedstability} to recover Theorem~\ref{thm:Bernstein}.

Consider next $\epsilon>0$ for which we study two cases. First, if $|f_i(\tv)|\ge p>0$ for all $\tv$ then we can choose $T=\infty$ in Theorem~\ref{thm:boundedstability} and obtain \eqref{eq:bern:052-03} because $|\Phi_i(\tv)|\le 1$ in \eqref{eq:bern:052}. The more difficult case is if we cannot find a $p$ for which $|f_i(\tv)|\ge p>0$ for all $\tv$. The following proof is based on a recursive argument similar to the one in \cite{Klebanov1986}. One novelty is applying Lemma~\ref{lemma:c.f.upperbound} to deal with random vectors.

To begin, suppose that the distributions of $\Xm_i$, $i=1,2$, are concentrated on two hyperplanes $\set S_i + \mv_i \subseteq \mathbb R^d$, where $\set S_i$ is a linear subspace of $\mathbb R^d$. Let $\tv$ be in the subspace $\set S_1^\perp$ orthogonal to $\set S_1$. We then have $f_1(\tv)=e^{j \tv^T \mv_1}$ and \eqref{eq:bern:052} gives $e^{-\tv^T \widehat \Qm \tv / 2} \ge 1-C(\epsilon)$. Suppose next that $\tv\notin\set S_2^\perp$ so that $|f_2(\tv)|<1$ and \eqref{eq:bern:052} gives $e^{-\tv^T \widehat \Qm \tv / 2} \le |f_2(\tv)|+ C(\epsilon)$. However, for sufficiently small $\epsilon$, this is a contradiction. We may thus assume that $\set S_1^\perp\subseteq\set S_2^\perp$ and hence $\set S_2 \subseteq \set S_1$. Repeating the argument for $\tv\in\set S_2^\perp$ we find that $\set S_1 \subseteq \set S_2$ and hence $\set S_1 = \set S_2$.

Next, if $\set S_1$ has dimension $d'$ with $d'<d$, then we may study the $d'$-dimensional distributions on this hyperplane\footnote{For example, if $d'=0$ then the distribution is a single point mass which is a Gaussian distribution with zero variance. As a more interesting example, if $d\ge2$ and the distribution is two distinct point masses, then the distribution is concentrated on a line, and we have $d'=1$.}. That is, we may as well choose $\Phi_i(\tv)$ so that $|\Phi_i(\tv)|=|f_i(\tv)|=1$ for $\tv\in\set S_1^\perp$ by selecting $\widehat \Qm$ such that $\tv^T \widehat \Qm \tv=0$ if $\tv\in\set S_1^\perp$.

So suppose the distributions of $\Xm_i$, $i=1,2$, are concentrated on the same hyperplane of dimension at least one. Since $f_i({\bf 0})=1$ and $f_i$ is continuous, we may choose $T>0$ sufficiently small such that the bounds of \eqref{eq:ushakov-theorem} and \eqref{eq:bern:002} are valid.
For example, we will be interested in $\|\tv\|=T/4$ for which there are positive $c,p$ such that
\begin{align}
    p < |f_i(\tv)| \le 1-\frac{c\, T^2}{16}.
\end{align}
Moreover, for sufficiently small $\epsilon$ we have $C(\epsilon) \le c\, T^2/32$ and \eqref{eq:bern:052} gives $|\Phi_i(\tv)| \le 1-c\, T^2/32$.

Define $r_{\epsilon,i}(\tv):=f_i(\tv)-\Phi_i(\tv)$ and consider the following steps based on \cite[Eq.~(13)]{Klebanov1986}:
\begin{align}
    & \left|r_{\epsilon,i}(2\tv)\right|
    \overset{(a)}{\le} \Big| f_i^2 |f_i|^2-\Phi_i^2|\Phi_i|^2 \Big| + 5\epsilon \nonumber\\
    & = \Big| \left(f_i-\Phi_i\right) \cdot \left( f_i |f_i|^2 + |f_i|^2 \Phi_i + f_i^* \Phi_i^2 \right) \nonumber \\
    & \qquad + \left( f_i-\Phi_i \right)^* \cdot \Phi_i^3 \Big|
    + 5\epsilon \nonumber \\
    & \le \left(|f_i|^3 + |f_i|^2 |\Phi_i| + |f_i| |\Phi_i|^2 + |\Phi_i|^3 \right)\cdot |r_{\epsilon,i}(\tv)| \,+\, 5\epsilon \nonumber\\
    & \le \left(|f_i|+|\Phi_i|\right)^3 |r_{\epsilon,i}(\tv)| \,+\, 5\epsilon \nonumber \\
    & \le \left(2|\Phi_i|+|r_{\epsilon,i}(\tv)|\right)^3 |r_{\epsilon,i}(\tv)| \,+\, 5\epsilon
    \label{eq:005-09}
\end{align}
where step $(a)$ follows from \eqref{eq:cf-Gauss2} and Lemma \ref{lemma:klebanov1} with $T=\infty$, and where we have written $f_i=f_i(\tv)$ and $\Phi_i=\Phi(\tv)$ for notational convenience.

Under the conditions of Theorem \ref{thm:boundedstability}, and using $|\Phi_i(\tv)|\le1$, we proved that
\begin{align}
    \left|r_{\epsilon,i}(\tv)\right| \le C(\epsilon),\quad \|\tv\| \in [0,T/2]
    \label{eq:005-10}
\end{align}
for $i=1,2$. Recall also that $|\Phi_i(c' \tv)|$ is non-increasing in $c'\ge0$ for any $\tv$. Now consider $\|\tv\|=T/4$ and define
\begin{align}
    \tv_0 := \argmax_{\tv:\: \tv\in\set S,\: \|\tv\|=T/4} \; |\Phi_i(\tv)|
\end{align}
and note that $|\Phi_i(\tv_0)|\le 1-c\, T^2/32<1$. Further define $C_0(\epsilon):=C(\epsilon)$ so that using \eqref{eq:005-09} we have
\begin{align}
    \left|r_{\epsilon,i}(\tv)\right| \le C_1(\epsilon), \quad  \|\tv\| \in [T/2,T]
    \label{eq:005-10-01}
\end{align}
where
\begin{align}
    C_1(\epsilon):=\left(2|\Phi_i(\tv_0)|+C_0(\epsilon)\right)^3 C_0(\epsilon)+5\epsilon.
\end{align}
By induction, we have
\begin{align}
    \left|r_{\epsilon,i}(\tv)\right| \le C_k(\epsilon), \;\;  \|\tv\| \in [2^{k-2}T,2^{k-1} T]
    \label{eq:005-10-01-01}
\end{align}
for $k\ge1$ where
\begin{align}
    C_k(\epsilon) := \left(2|\Phi_i(2^{k-1}\tv_0)|+C_{k-1}(\epsilon)\right)^3 C_{k-1}(\epsilon)  +5\epsilon.
    \label{eq:005-10-01-02}
\end{align}
Note that since $C_0(\epsilon)= C(1)\,\epsilon$, the error term $C_k(\epsilon)$ is a polynomial in $\epsilon$ with constant coefficient zero and all other coefficients positive.

Next, since $|\Phi_i(2^{k-1}\tv_0)|$ is non-increasing with $k$, if $C_k(\epsilon)\le C_{k-1}(\epsilon)$ then by induction $C_\ell(\epsilon)\le C_{k-1}(\epsilon)$ for all $\ell\ge k$.\footnote{For example, for $d=1$ we have $C_0(\epsilon)=1440\epsilon/p^4$. Thus, if the cubic term satisfies $(2|\Phi_i(\tv_0)|+C_0(\epsilon))^3\le(1435/1440)p^4$ then $C_1(\epsilon)\le C_0(\epsilon)$ and also $C_\ell(\epsilon)\le C_0(\epsilon)$ for all $\ell\ge0$.} More generally, if $C_k(\epsilon)\le C_i(\epsilon)$ for any $k>i$ then by induction $C_{k+\ell}(\epsilon)\le C_{i+\ell}(\epsilon)$ for $\ell\ge 0$. We use these bounds to complete the proof.

Recall that $|\Phi_i(\tv_0)|<1$ and $|\Phi_i(2^{k-1}\tv_0)|=|\Phi_i(\tv_0)|^{4^{k-1}}$ (see~\eqref{eq:cf-Gauss2}). Thus, for any $x>0$ there is a smallest positive $k$ such that  $2|\Phi_i(2^{k-1}\tv_0)|\le x$, and this $k$ is independent of $\epsilon$.\footnote{This argument does not work if $|f_i(\tv_0)|=1$ because then $\epsilon\rightarrow0$ requires $|\Phi_i(\tv_0)|\rightarrow1$.\label{footnote:f1}} For $x<2^{-1/3}\approx 0.7937$, we can therefore determine $k$ and also a sufficiently small $\epsilon$ such that
\begin{align}
    (2|\Phi_i(2^{k-1}\tv_0)|+C_{k-1}(\epsilon))^3\le\frac{1}{2}.
    \label{eq:005-10-01-02a}
\end{align}
Now, if $C_k(\epsilon)\le C_{k-1}(\epsilon)$ then we have the desired result \eqref{eq:bern:052-03} with
\begin{align}
    \tilde C = \max_{0\le i\le k-1} C_i(1).
    \label{eq:C-bound}
\end{align}
But if $C_k(\epsilon)> C_{k-1}(\epsilon)$ then \eqref{eq:005-10-01-02} and \eqref{eq:005-10-01-02a} give the bound $C_{k-1}(\epsilon) < 10 \epsilon \le C_0(\epsilon)$. We thus again have \eqref{eq:C-bound}.

\subsection{Proof of Theorem~\ref{thm:entropies}}
\label{subsec:stability3-proof}
The fourth moments of $\Ym_{g,i}$ are bounded if the second moments of $\Ym_{g,i}$ are bounded \cite[p.~148]{papoulis2002probability}. Thus, Theorem~\ref{thm:unboundedstability}, Lemma~\ref{lemma:pointwise-cf-to-pdf}, and Lemma~\ref{lemma:pointwise-pdf-to-L1} with $\Ym_q=\Ym_{g,i}$ give \eqref{eq:bern:052-03-entropy-Q} with $B(\epsilon)=B_3(\tilde C \epsilon)$.

To prove \eqref{eq:bern:052-03-entropy}, we apply Lemma~\ref{lemma:stability-entropy} with (see~\eqref{eq:ac-density})
\begin{align}
    & \underset{\yv\in\mathbb R^d}{\text{ess sup}} \; p_{\Ym_i}(\yv) \le \phi_{\Zm_i}({\bf 0}) = \det\left(2\pi \Qm_{\Zm_i}\right)^{-1/2} \label{eq:appd-bound7a} \\
    & \underset{\yv\in\mathbb R^d}{\text{ess sup}} \; \phi_i(\yv) \le \det\left(2\pi \Qm_{\Zm_i}\right)^{-1/2} + B_1(\tilde C \epsilon)
    \label{eq:appd-bound7b}
\end{align}
and $\alpha=2$ for which \eqref{eq:bern:052-03-entropy-Q} gives
\begin{align}
    \E{\|\Ym_{g,i}\|_2^2} & \le \E{\|\Ym_i\|_2^2} + B_3(\tilde C \epsilon) \, d .
    \label{eq:appd-bound7d}
\end{align}
The random vectors with p.d.f.s $p_{\Ym_i}$ and $\phi_i$ are thus in the class $(2,\nu,m)-\mathcal{AC}^d$ of Lemma~\ref{lemma:stability-entropy} where for sufficiently small $\epsilon$ we may choose
\begin{align*}
    m = \det\left(2\pi \Qm_{\Zm_i}\right)^{-1/2} + 1, \quad
    \nu = \max_{i=1,2} \, \E{\|\Ym_i\|_2^2} + 1.
\end{align*}
Lemmas~\ref{lemma:pointwise-pdf-to-L1} and~\ref{lemma:stability-entropy} now give
\begin{align}
    & |h(p_{\Ym_i}) - h(\phi_i)| \nonumber \\
    & \le \| p_{\Ym_i} - \phi_i \| \cdot \big( c_1 - c_2 \log \| p_{\Ym_i} - \phi_i\| \big) \nonumber \\
    & \le B_2(\epsilon) \big( c_1 - c_2 \log B_2(\epsilon) \big) := B_4(\epsilon)\label{eq:B4-Definition}
\end{align}
where 
\begin{align}
    c_1 = \frac{d}{2}\left|\log \frac{4 \nu \pi e}{d}\right| + \log \frac{me^2}{2} +1, \quad
    c_2 = \frac{d}{2} + 2
\end{align}
and where we choose $B_2(\epsilon)$ large enough to be valid for both $i=1,2$. We thus have $B_4(\epsilon)\rightarrow0$ as $\epsilon\rightarrow0$ and $B_4(\epsilon)=0$ if $\epsilon=0$. Finally, we may choose $B(\epsilon)=\max(B_3(\tilde C \epsilon),B_4(\epsilon))$.

\subsection{Discussion}
\label{subsec:discussion}

\subsubsection{Extending Existing Results}
Theorems~\ref{thm:boundedstability} and~\ref{thm:unboundedstability} include Bernstein's Theorem as a special case, and they extend its stability from scalars to vectors. Note that~\cite{Klebanov1986} treats $d=1$ but does not prove a common variance for $X_1$ and $X_2$. Moreover, the scaling of Theorems~\ref{thm:boundedstability} and~\ref{thm:unboundedstability}  is proportional to $\epsilon$ which is the best possible in general; see~\cite{Klebanov1986} and~\eqref{eq:uniform-metric4}.

\subsubsection{Common Covariance Matrix}
We could have chosen different covariance matrices $\widehat \Qm_i$ for $i=1,2$ in the proof and statement of Theorem~\ref{thm:boundedstability}. The bound \eqref{eq:bern:010-01a} then ensures that $\widehat \Qm_1$ is close to $\widehat \Qm_2$. In contrast, choosing $\widehat \Qm_1=\widehat \Qm_2$ requires reducing $\epsilon$ to maintain the same approximation precision. One advantage of the former approach is that one can treat each covariance matrix separately, e.g., one can select the $\widehat \Qm_i$ so that $\tv^T \widehat \Qm_i \tv=0$ if $\tv\in\set S_i^\perp$ without requiring $\set S_1=\set S_2$.

\subsubsection{Potential Pitfalls}
The proof of Theorem~\ref{thm:unboundedstability} had two potential problematic cases permitted by perturbing Cauchy's functional equation:
\begin{itemize}
    \item $|f_i(\tv)|=1$ but $\Phi_i$ in Theorem~\ref{thm:boundedstability} has $\tv^T \widehat \Qm \tv>0$;
    \item $|f_i(\tv)|<1$ but $\Phi_i$ in Theorem~\ref{thm:boundedstability} has $\tv^T \widehat \Qm \tv=0$.
\end{itemize}
These two cases could have invalidated Theorem~\ref{thm:unboundedstability} because, in either case, $f_i$ is not stable with respect to the chosen $\Phi_i$. Fortunately, however, Lemma~\ref{lemma:c.f.upperbound} lets one resolve both cases. The first case must correspond to a degenerate distribution for which one can choose $\widehat \Qm$ to have zero eigenvalues in the subspace orthogonal to the non-degenerate hyperplane. The second case can be avoided because Lemma~\ref{lemma:c.f.upperbound} lets one choose a sufficiently small $\epsilon$ for which $\tv^T \widehat \Qm \tv>0$.

\subsubsection{$\delta$-Dependent $\Xm_1$ and $\Xm_2$}
Theorems~\ref{thm:boundedstability}-\ref{thm:entropies} generalize to $\Xm_1$ and $\Xm_2$ that are $\delta$-dependent for $\delta>0$ by replacing $\epsilon$ with $\epsilon+4\delta$. We state this formally as an extension of Theorem~\ref{thm:unboundedstability}.

\begin{theorem}\label{thm:unboundedstability-general}
Suppose $\Xm_1$ and $\Xm_2$ are $\delta$-dependent random vectors, and $\Xm_1+\Xm_2$ and $\Xm_1-\Xm_2$ are $\epsilon$-dependent.
Then for all $\epsilon+4\delta$ below some positive threshold, for all $\tv \in {\mathbb R}^d$, and for $i=1,2$ we have
\begin{align}
    \left|f_i(\tv)-\Phi_i(\tv)\right| \le \tilde C (\epsilon+4\delta) \label{eq:bern:052-03-g}
\end{align}
for the Gaussian c.f.s \eqref{eq:cf-Gauss-i}, and
for a constant $\tilde C$ independent of $\epsilon+4\delta$ and $\tv$. In particular, if $\delta=0$ then we recover Theorem~\ref{thm:unboundedstability}, and if $\epsilon=\delta=0$ then $\Xm_1$ and $\Xm_2$ are Gaussian with the same covariance matrix.
\end{theorem}
\begin{proof}
Consider the proof of Theorem \ref{thm:boundedstability} and assume that $\Xm_1$ and $\Xm_2$ are $(\delta, 2T)$-dependent. For the three c.f.s on the right hand sides of \eqref{eq:004-a} and \eqref{eq:004-b} we have
\begin{align}
    \left| f_{\Xm_1,\Xm_2}(\tv_1+\tv_2,\tv_1-\tv_2) \qquad \qquad \right. \nonumber \\
    \quad \left. - f_{\Xm_1}(\tv_1+\tv_2)f_{\Xm_2}(\tv_1-\tv_2) \right| & \le \delta \label{eq:004-a-01}\\
    \left| f_{\Xm_1,\Xm_2}(\tv_1,\tv_1) - f_{\Xm_1}(\tv_1)f_{\Xm_2}(\tv_1)\right| & \le \delta \label{eq:004-b-01}\\
    \left| f_{\Xm_1,\Xm_2}(\tv_2,-\tv_2) - f_{\Xm_1}(\tv_2)f_{\Xm_2}(-\tv_2) \right| & \le \delta \label{eq:004-b-02}
\end{align}
for all $\|\tv_1\| \le T$, $\|\tv_2\| \le T$. Similar to \eqref{eq:bern:004} and \eqref{eq:bern:004-r}, by combining \eqref{eq:004-a}, \eqref{eq:004-b} with \eqref{eq:004-a-01}-\eqref{eq:004-b-02} we have
\begin{align}
    &f_1(\tv_1+\tv_2) f_2(\tv_1-\tv_2)\nonumber\\
    &= f_1(\tv_1)f_1(\tv_2) f_2(\tv_1)f_2(-\tv_2)+r_{\epsilon, \delta}(\tv_1, \tv_2)
        \label{eq:bern:004-general}
\end{align}
where
\begin{align}
    |r_{\epsilon, \delta}(\tv_1, \tv_2)| \le \epsilon+4\delta \text{ and }
    \|\tv_i\| \le T, \; i=1, 2.
    \label{eq:bern:004-r-general}
\end{align}
The remaining steps follow by replacing $\epsilon \leftarrow \epsilon+4\delta$ in the proofs of Theorems \ref{thm:boundedstability} and \ref{thm:unboundedstability}.
\end{proof}

\section{Soft Doubling for AGN Channels}
\label{sec:soft-doubling}

This section shows how to combine the stability of Bernstein's theorem with the doubling argument in~\cite{Geng-Nair-IT14} to obtain a soft doubling argument that does not require the existence of maximizers. We consider point-to-point channels, product channels, and broadcast channels with AGN that have applications to cellular wireless networks~\cite{Kim-Kramer-C22}.

\subsection{Point-to-Point Channels}
\label{sec:point2point}
An AGN channel has output
\begin{align}
    \Ym=\Xm+\Zm
    \label{eq:AGN-model}
\end{align}
where $\Xm, \Ym, \Zm \in \mathbb{R}^d$ and $\Zm \sim \Nc({\bf 0},\Qm_{\Zm})$ is independent of $\Xm$. Consider the optimization problem: 
\begin{align}
  V(\Qm) \coloneqq \sup_{\Xm:\,\E{\Xm\Xm^T}\preceq\Qm} I(\Xm;\Ym)
  \label{eq:001}.
\end{align}
We use Theorem~\ref{thm:entropies} and a soft doubling argument to prove the following known result.

\begin{proposition}
\label{prop:agn}
For the AGN channel \eqref{eq:AGN-model} we have
\begin{align}
    I(\Xm;\Ym) 
    & \le \frac{1}{2} \log\frac{\det \left( \Qm_{\Xm} + \Qm_{\Zm} \right)}{\det \Qm_{\Zm}}
    \label{eq:022-01}
\end{align}
with equality if $\Xm$ is Gaussian.
\end{proposition}
\begin{proof}
Equality holds in \eqref{eq:022-01} if $\Xm$ is Gaussian, so it remains to prove the inequality. Note that we may assume $\E{\Xm}={\bf 0}$ because $I(\Xm;\Ym)$ does not depend on translation of $\Xm$.

Now consider $\Ym_1=\Xm_1+\Zm_1$ and $\Ym_2=\Xm_2+\Zm_2$, where $\Xm_1, \Xm_2 \sim P_{\Xm}$ and $\Zm_1, \Zm_2 \sim P_{\Zm}$ are mutually independent. Further, define the vectors
\begin{align}
    & \Xm_{+}:=\frac{1}{\sqrt{2}}(\Xm_1+\Xm_2), &&
    \Xm_{-}:=\frac{1}{\sqrt{2}}(\Xm_1-\Xm_2).\label{eq:010}
\end{align}
Note that $\E{\Xm_{+}\Xm_{+}^T}\preceq\Qm$ and $\E{\Xm_{-}\Xm_{-}^T}\preceq\Qm$.
Also, define 
\begin{align}
    & \Ym_{+} :=\frac{1}{\sqrt{2}}(\Ym_1+\Ym_2), && \Ym_{-}:=\frac{1}{\sqrt{2}}(\Ym_1-\Ym_2) \label{eq:010-01} \\
    &\Zm_{+} :=\frac{1}{\sqrt{2}}(\Zm_1+\Zm_2), && \Zm_{-}:=\frac{1}{\sqrt{2}}(\Zm_1-\Zm_2)\label{eq:010-02}
\end{align}
so that $\Ym_{+}=\Xm_{+}+\Zm_{+}$ and $\Ym_{-}=\Xm_{-}+\Zm_{-}$ where the noise vectors $\Zm_{+}, \Zm_{-} \sim P_{\Zm}$ are independent.

Now suppose $\epsilon>0$ and
\begin{align}
    I(\Xm;\Ym) = V(\Qm_{\Xm})-\epsilon.
    \label{eq:007} 
\end{align}
We then have
\begin{align}
2V(\Qm_{\Xm}) & = I(\Xm_1;\Ym_1)+I(\Xm_2;\Ym_2)+2\epsilon \nonumber \\
&= I(\Xm_1\Xm_2;\Ym_1\Ym_2)+2\epsilon \nonumber \\
& =I(\Xm_{+}\Xm_{-};\Ym_{+}\Ym_{-})+2\epsilon \nonumber \\
& \overset{(a)}{=} \underbrace{I(\Xm_{+};\Ym_{+})}_{\displaystyle \le V(\Qm_{\Xm})}+\underbrace{I(\Xm_{-};\Ym_{-})}_{\displaystyle \le V(\Qm_{\Xm})}-I(\Ym_{+};\Ym_{-})+2\epsilon \nonumber \\
& \le 2V\left( \Qm_{\Xm} \right) - I(\Ym_{+};\Ym_{-}) + 2\epsilon
\label{eq:015}
\end{align}
where step $(a)$ follows by
\begin{align*}
    p(\yv_{+},\yv_{-}|\xv_{+},\xv_{-}) = p(\yv_{+}|\xv_{+})\,p(\yv_{-}|\xv_{-})
\end{align*}
for all $\xv_{+},\xv_{-},\yv_{+},\yv_{-}$. Lemma \ref{lemma:dt-to-mutualinfo} and \eqref{eq:015} give
\begin{align}
    & \left| f_{\Ym_{+}, \Ym_{-}}(\tv_1, \tv_2) - f_{\Ym_{+}}(\tv_1) f_{\Ym_{-}}(\tv_2) \right| \nonumber \\
    & \quad \le \sqrt{2 I(\Ym_{+};\Ym_{-})} \le 2\sqrt{\epsilon}
    \label{eq:018}
\end{align}
so $\Ym_{+}$ and $\Ym_{-}$ are $(2\sqrt{\epsilon})$-dependent.

For sufficiently small $\epsilon$, Theorem \ref{thm:entropies} gives
\begin{align}
    I(\Xm;\Ym) & = h(\Ym_1)-h(\Zm) \nonumber \\
    & \le h(\Ym_{g,1}) - \frac{1}{2} \log \det (2 \pi \, \Qm_{\Zm}) + B\left( 2\sqrt{\epsilon} \right) \nonumber \\
    & = \frac{1}{2} \log\frac{\det \Qm_{\Ym_{g,1}}}{\det \Qm_{\Zm}} + B\left( 2\sqrt{\epsilon} \right)
    \label{eq:018a}
\end{align}
where for some mean vector $\mv_{g,1}$ we have
\begin{align}
    \Qm_{\Ym_{g,1}}
    & = \E{\Ym_{g,1} \Ym_{g,1}^T} - \mv_{g,1} \mv_{g,1}^T \nonumber \\
    & \preceq \E{\Ym_1 \Ym_1^T} + B\left( 2\sqrt{\epsilon} \right) \Id_d \nonumber \\
    & \overset{(a)}{=} \Qm_{\Xm} + \Qm_{\Zm} + B\left( 2\sqrt{\epsilon} \right) \Id_d
    \label{eq:018b}
\end{align}
and step $(a)$ follows by $\E{\Xm}={\bf 0}$. Moreover, the $\epsilon$ in~\eqref{eq:007} can be chosen as close to zero as desired because $V(\Qm_{\Xm})$ is a supremum. Finally, observe that if $I(\Xm;\Ym)\le J+\epsilon$ for all $\epsilon>0$ then $I(\Xm;\Ym) \le J$.
\end{proof}

Note that the proof of the inequality~\eqref{eq:022-01} does not require the existence of a maximizer. Also, a maximizer may not be unique. Proving existence and uniqueness is interesting but not crucial for the communications problem.

\subsection{Product Channels}
\label{sec:AGN-Product}
The proof of Proposition~\ref{prop:agn} uses an AGN product channel with two outputs $\Ym_1=\Xm_1+\Zm_1$ and $\Ym_2=\Xm_2+\Zm_2$
where $\Zm_1$ and $\Zm_2$ have the same covariance matrix. More generally, suppose the covariance matrices are different, i.e., consider the AGN product channel
\begin{align}
    \Ym_1 & = \Xm_1 + \Zm_1 \label{eq:product-channel-1} \\
    \Ym_2 & = \Xm_2 + \Zm_2 \label{eq:product-channel-2}
\end{align}
where $\Zm_1\sim \Nc({\bf 0},\Qm_{\Zm_1})$ and $\Zm_2\sim \Nc({\bf 0},\Qm_{\Zm_2})$ are non-degenerate and independent, and $(\Xm_1,\Xm_2)$ is independent of $(\Zm_1,\Zm_2)$. Since the noise is non-degenerate, this channel is equivalent to the AGN product channel considered in~\cite[Sec.~I.A]{Geng-Nair-IT14}, namely\footnote{We replace the $\Gm_{11},\Gm_{22}$ in~\cite[Sec.~I.A]{Geng-Nair-IT14} with $\Gm_1,\Gm_2$.}
\begin{align}
    \Ym_{11} & = \Gm_1 \Xm_1 + \Zm_{11} \label{eq:product-channel-1a} \\
    \Ym_{22} & = \Gm_2 \Xm_2 + \Zm_{22} \label{eq:product-channel-2a}
\end{align}
where $\Gm_{1}=\Qm_{\Zm_1}^{-1/2}$, $\Gm_{2}=\Qm_{\Zm_2}^{-1/2}$, the noise vectors $ \Zm_{11}, \Zm_{22} \sim \Nc({\bf 0},\Id_d)$ are independent, and $(\Xm_1,\Xm_2)$ is independent of $(\Zm_{11},\Zm_{22})$.

We need a statement similar to
Proposition~\cite[Prop.~2]{Geng-Nair-IT14} which states that $\Ym_{11},\Ym_{22}$ are independent if and only if $\Xm_1,\Xm_2$ are independent. This is important because $\Xm_1,\Xm_2$ are required to be independent to apply Bernstein's theorem. So the question is whether a similar result holds for $\epsilon$-dependence.

Unfortunately, this does not seem to work. Observe that $f_{\Gm\Xm}(\tv)=f_{\Xm}(\Gm^T\tv)$ and therefore
\begin{align}
    & \left| f_{\Ym_{11},\Ym_{22}}(\tv_1,\tv_2) - f_{\Ym_{11}}(\tv_1)f_{\Ym_{22}}(\tv_2) \right| \nonumber \\
    & = \left| f_{\Xm_1,\Xm_2}\left( \Gm_1^T\tv_1,\Gm_2^T\tv_2 \right) - f_{\Xm_1}\left( \Gm_1^T\tv_1 \right) 
    f_{\Xm_2}\left( \Gm_2^T\tv_2 \right) \right| \nonumber \\
    & \quad \cdot |\Phi_{\Zm_{11}}(\tv_{11})| \cdot |\Phi_{\Zm_{22}}(\tv_2)|
    \label{eq:XY-epsilon-dependence}
\end{align}
where $\Phi_{\Zm_{ii}}(\tv)=e^{-\frac{1}{2} \|\tv\|^2}\le 1$ for $i=1,2$. Thus, if $\Xm_1,\Xm_2$ are $\epsilon$-dependent then $\Ym_{11},\Ym_{22}$ are $\epsilon$-dependent. However, the converse statement, namely that if $\Ym_{11},\Ym_{22}$ are $\epsilon$-dependent then $\Xm_1,\Xm_2$ are $\epsilon$-dependent, is not valid in general.\footnote{The robust $\epsilon$-dependence in Appendix~\ref{appendix:robust} does apply in both directions: Lemma~\ref{lemma:robustly-epsilon-dependent} states that $\Ym_{11},\Ym_{22}$ are robustly  $\epsilon$-dependent if and only if $\Xm_1,\Xm_2$ are robustly $\epsilon$-dependent.}

We therefore take a different approach. Consider the noise $\Zm_1'\sim \Nc({\bf 0},\Qm_{\Zm_1})$ and $\Zm_2'\sim \Nc({\bf 0},\Qm_{\Zm_2})$ and suppose $\Xm_1,\Xm_2,\Zm_1,\Zm_2,\Zm_1',\Zm_2'$ are mutually independent. Define the physically degraded channels
\begin{align}
    \widetilde \Ym_1 &
    = \Ym_1 + \Zm_2'
    = \Xm_1 + \widetilde \Zm_1 \label{eq:product-channel-1b} \\
    \widetilde \Ym_2 &
    = \Ym_2 + \Zm_1'
    = \Xm_2 + \widetilde \Zm_2 \label{eq:product-channel-2b}
\end{align}
where $\widetilde \Zm_1 := \Zm_1 + \Zm_2'$ and $\widetilde \Zm_2 := \Zm_2 + \Zm_1'$ are independent and have the same covariance matrix $\Qm_{12}=\Qm_{\Zm_1}+\Qm_{\Zm_2}$. Observe that $\widetilde \Ym_1,\widetilde \Ym_2$ are independent.

Next, consider $\Xm_{+},\Xm_{-}$ as in~\eqref{eq:010} and define
\begin{align}
    & \Ym_{1,+} = \Xm_{+} + \Zm_{1}, \quad
    && \Ym_{2,-} = \Xm_{-} + \Zm_{2} \\
    & \Ym_{11,+} = \Gm_1 \Xm_{+} + \Zm_{11},
    &&
    \Ym_{22,-} = \Gm_2 \Xm_{-} + \Zm_{22} \\
    & \widetilde \Ym_{+} = \frac{1}{\sqrt{2}} (\widetilde \Ym_1 + \widetilde \Ym_2),
    && \widetilde \Ym_{-} = \frac{1}{\sqrt{2}} (\widetilde \Ym_1 - \widetilde \Ym_2) .
    \label{eq:tilde-Y-plus}
\end{align}
We have
\begin{align}
    I(\Ym_{11,+} ; \Ym_{22,-} )
    \overset{(a)}{=} I(\Ym_{1,+} ; \Ym_{2,-})
    \overset{(b)}{\ge} I(\widetilde \Ym_{+}; \widetilde \Ym_{-} )
    \label{eq:I-bc-bound}
\end{align}
where step $(a)$ follows because $\Ym_{11,+}=\Gm_1 \Ym_{1,+}$ and $\Ym_{22,-}=\Gm_2 \Ym_{2,-}$ for invertible $\Gm_1,\Gm_2$, and step $(b)$ follows because $\widetilde \Ym_{+},\widetilde \Ym_{-}$ are degraded versions of $\Ym_{1,+},\Ym_{2,-}$. 
Note that one cannot write the distributions of $\Ym_{11,+}$ and $\Ym_{22,-}$ as the distributions of the respective $\frac{1}{\sqrt{2}}(\Ym_{11}+\Ym_{22})$ and $\frac{1}{\sqrt{2}}(\Ym_{11}-\Ym_{22})$ since $\Gm_1$ and $\Gm_2$ are different in general. This is why we introduced $\widetilde \Ym_1,\widetilde \Ym_2$ for which we can write the forms \eqref{eq:tilde-Y-plus}. 

We now use the $\epsilon$-dependence of $\Ym_{11,+},\Ym_{22,-}$ to show that $\Xm_1,\Xm_2$ are approximately Gaussian without considering $\Xm_{+},\Xm_{-}$ directly. We refine the proof of Theorem~\ref{thm:boundedstability} to obtain modified versions of Theorems~\ref{thm:unboundedstability} and~\ref{thm:entropies} for the general AGN product channel of interest. 

\begin{theorem}\label{thm:unboundedstability-2}
Consider the AGN product channel~\eqref{eq:product-channel-1a}-\eqref{eq:product-channel-2a} and suppose $I(\Ym_{11,+} ; \Ym_{22,-})\le\epsilon$.
Then for all $\epsilon$ below some positive threshold, for all $\tv \in {\mathbb R}^d$, and for $i=1,2$ we have
\begin{align}
    \left|f_{\Xm_i}(\tv)-\Phi_i(\tv)\right| \le \tilde C \sqrt{2\epsilon} \label{eq:bern:052-03-2}
\end{align}
for the Gaussian c.f.s \eqref{eq:cf-Gauss-i}, and for a constant $\tilde C$ independent of $\epsilon$ and $\tv$.
\end{theorem}
\begin{proof}
The bound~\eqref{eq:I-bc-bound} gives $I(\widetilde \Ym_{+} ; \widetilde \Ym_{-}) \le \epsilon$ and Lemma~\ref{lemma:dt-to-mutualinfo} implies that $\widetilde \Ym_{+}$ and $\widetilde \Ym_{-}$ are $\sqrt{2\epsilon}$-dependent. We now apply the steps of the proof of Theorem~\ref{thm:boundedstability} with $\widetilde \Ym_1,\widetilde \Ym_2$ replacing $\Xm_1,\Xm_2$. Moreover, observe that $g_{\widetilde \Ym_i}(\tv)=g_{\Xm_i}(\tv)-\frac{1}{2}\tv^{\mathrm{T}}\Qm_{12}\tv$ so that \eqref{eq:bern:045} becomes
\begin{align}
    g_{\Xm_1}(\tv)& = j\tv^{\mathrm{T}}\widehat\mv_{1,I}-\frac{1}{2}\tv^{\mathrm{T}}\left(\widetilde\Qm_R - \Qm_{12} \right)\tv + R^{(6)}_{\sqrt{2\epsilon}}(\tv).
    \label{eq:bern:045-modified}
\end{align}
Following the same steps starting with \eqref{eq:bern:046-cov}, we may write \eqref{eq:bern:048} as
\begin{align}
    g_{\Xm_1}(\tv)=j\tv^{\mathrm{T}}\widehat\mv_{1,I}-\frac{1}{2}\tv^{\mathrm{T}}\widehat\Qm_R\tv+R^{(7)}_{\sqrt{2\epsilon}}(\tv)
    \label{eq:bern:048-modified}
\end{align}
where $\widehat\Qm_R\succeq{\bf 0}$ is a covariance matrix that is perturbed version of $\widetilde\Qm_R - \Qm_{12}$ rather than $\widetilde\Qm_R$. Continuing as for the proofs of Theorems~\ref{thm:boundedstability} and~\ref{thm:unboundedstability}, we obtain~\eqref{eq:bern:052-03-2}.
\end{proof}

Observe that the $\Xm_1,\Xm_2$ in Theorem~\ref{thm:unboundedstability-2} do not necessarily have densities. However, the following result shows that once the $\Xm_1,\Xm_2$ are characterized as approximately Gaussian with the same covariance matrix $\widehat \Qm$, then the $\Ym_{11},\Ym_{22}$ are approximately Gaussian with covariance matrices based on $\widehat \Qm$. Moreover, both $\Ym_{11}$ and $\Ym_{22}$ have densities.

\begin{corollary} \label{corr:unboundedstability-2}
Under the conditions of Theorem~\ref{thm:unboundedstability-2}, we have
\begin{align}
    \left|f_{\Ym_{ii}}(\tv)-\widetilde \Phi_i(\tv)\right| \le \tilde C \sqrt{2\epsilon} \label{eq:bern:052-03-3}
\end{align}
for all $\tv \in {\mathbb R}^d$,
for Gaussian c.f.s $\widetilde \Phi_i(\tv)$ with covariance matrices $\Gm_i \widehat \Qm \Gm_i^T + \Id_d$, $i=1,2$, and for the $\tilde C$ in \eqref{eq:bern:052-03-2}.
\end{corollary}
\begin{proof}
Using \eqref{eq:bern:052-03-2}, we have
\begin{align}
    & \left|f_{\Ym_{ii}}(\tv)-\Phi_i\left(\Gm_i^T\tv\right) \Phi_{\Zm_{ii}}(\tv) \right| \nonumber \\
    & = \underbrace{\left|f_{\Xm_i}\left(\Gm_i^T\tv\right)-\Phi_i\left(\Gm_i^T\tv\right)\right|}_{\displaystyle \le \tilde C \sqrt{2\epsilon}} \cdot \underbrace{|\Phi_{\Zm_i}(\tv) |}_{\displaystyle \le 1}
\end{align}
where for some mean vectors $\widehat \mv_i$ we have (see~\eqref{eq:cf-Gauss-i})
\begin{align}
    \Phi_i\left(\Gm_i^T\tv\right) \Phi_{\Zm_{ii}}(\tv) = e^{\tv^T \Gm_i (j \widehat \mv_i - \frac{1}{2}\widehat \Qm \, \Gm_i^T \tv)} \cdot e^{- \frac{1}{2} \| \tv \|^2}.
    \label{eq:Gauss-Yi}
\end{align}
The covariance matrix of \eqref{eq:Gauss-Yi} is $\Gm_i \widehat \Qm \Gm_i^T + \Id_d$.
\end{proof}

\begin{theorem}\label{thm:entropies-2}
Consider the AGN product channel~\eqref{eq:product-channel-1a}-\eqref{eq:product-channel-2a} and suppose $\Xm_1,\Xm_2$ have finite second moments and $I(\Ym_{11,+} ; \Ym_{22,-})\le\epsilon$.
Then for all $\epsilon$ below some positive threshold and for $i=1,2$ and $\epsilon'=\sqrt{2\epsilon}$ we have
\begin{align}
    \left|h(\Ym_{ii})-h(\Ym_{g,i})\right| \le B\left( \epsilon' \right) \label{eq:bern:052-03-entropy-2}
\end{align}
where $\Ym_{g,i} \sim \Nc\Big(\mv_i,\Gm_i \widehat \Qm \Gm_i^T + \Id_d\Big)$ and
\begin{align}
    \E{\Ym_{g,i} \Ym_{g,i}^T}
    \preceq \E{\Ym_{ii} \Ym_{ii}^T} + B\left( \epsilon' \right) \Id_d
    \label{eq:bern:052-03-entropy-Q-Y}
\end{align}
where $B(\epsilon')\rightarrow 0$ as $\epsilon'\rightarrow 0$ and $B(\epsilon')=0$ if $\epsilon'=0$.
\end{theorem}
\begin{proof}
Use the same steps as for the proof of Theorem~\ref{thm:entropies} in Sec.~\ref{subsec:stability3-proof} but for \eqref{eq:bern:052-03-3}. For example, Lemmas~\ref{lemma:pointwise-cf-to-pdf} and~\ref{lemma:pointwise-pdf-to-L1} give \eqref{eq:bern:052-03-entropy-Q-Y} with $B(\epsilon')=B_3(\tilde C \epsilon')$.

\end{proof}

\subsection{Broadcast Channels}
\label{sec:broadcast}

The two-receiver AGN broadcast channel has
\begin{align}
    \Ym_1 & = \Gm_1 \Xm + \Zm_1 \label{eq:bc-1} \\
    \Ym_2 & = \Gm_2 \Xm + \Zm_2 \label{eq:bc-2}
\end{align}
where $\Gm_1,\Gm_2$ are invertible, and $\Zm_1, \Zm_2 \sim \Nc({\bf 0},\Id_d)$ are independent. Note that the input $\Xm$ is common to both sub-channels and $\Zm_1, \Zm_2$ have the same covariance matrix. Define the expressions (see~\cite{Geng-Nair-IT14})
\begin{align}
    s_\lambda(\Xm) & := I(\Xm;\Ym_1) - \lambda I(\Xm;\Ym_2) \\
    s_\lambda(\Xm|\Vm) & := I(\Xm;\Ym_1|\Vm) - \lambda I(\Xm;\Ym_2|\Vm) \\
    S_\lambda(\Xm) & := \sup_{p(\vv|\xv):\,\Vm - \Xm - \Ym_1\Ym_2} s_\lambda(\Xm|\Vm) \\
    V_\lambda(\Qm) & := \sup_{ \Xm:\,\E{\Xm\Xm^T}\preceq\Qm} S_\lambda(\Xm)
    \label{eq:VlQ}
\end{align}
where $S_\lambda(\Xm)$ is the upper concave envelope of $s_\lambda(\Xm)$ as a function of $p(\xv)$. We will need the following results concerning $s_\lambda(\Xm)$ and $s_\lambda(\Xm|\Vm)$.

\begin{lemma} \label{lemma:s_lambda-bound}
If $\lambda\ge1$ then
\begin{align}
    s_{\lambda}(\Xm) \le \frac{1}{2} \log \frac{\det\left(\Qm_1' + \Qm_2' \right)}{\det \Qm_1'} \label{eq:s1}
\end{align}
where $\Qm_k'=(\Gm_k^T \Gm_k)^{-1}$ for $k=1,2$.
\end{lemma}
\begin{proof}
Let $\Zm_k'=\Gm_k^{-1}\Zm_k$ so that $\Qm_k'=(\Gm_k^T \Gm_k)^{-1}$ is the covariance matrix of $\Zm_k'$ for $k=1,2$. We have
\begin{align}
    s_\lambda(\Xm) 
    & \le I(\Xm ; \Xm + \Zm_1') - I(\Xm ; \Xm + \Zm_2') \nonumber \\
    & = \left[h(\Xm + \Zm_1') - h(\Xm + \Zm_2')\right] - \frac{1}{2} \log \frac{\det\Qm_1'}{\det \Qm_2'} \nonumber \\
    & \le \underbrace{\left[h(\Xm + \Zm_1' + \Zm_2') - h(\Xm + \Zm_2')\right]}_{\displaystyle = I(\Xm + \Zm_1' + \Zm_2' ; \Zm_1')} - \frac{1}{2} \log \frac{\det\Qm_1'}{\det \Qm_2'} \nonumber \\
    & \le I(\Zm_1' + \Zm_2' ; \Zm_1') - \frac{1}{2} \log \frac{\det\Qm_1'}{\det \Qm_2'}
    \label{eq:s1a}
\end{align}
and evaluating \eqref{eq:s1a} gives \eqref{eq:s1}.
\end{proof}

\begin{lemma} \label{lemma:Caratheodory-bound}
For every pair $(\Vm,\Xm)$ there is a pair $(\Vm',\Xm')$ with $s_\lambda(\Xm'|\Vm')=s_\lambda(\Xm|\Vm)$ and $\E{\Xm' (\Xm')^T}=\E{\Xm \Xm^T}$, and where $\Vm'$ has alphabet of cardinality $d(d+1)/2+1$.
\end{lemma}
\begin{proof}
The result follows by the support lemma in \cite[Lemma 15.4]{csiszar2011information} that we rephrase with our notation (see also~\cite[pp.~2099-2100]{Geng-Nair-IT14}). Let $\Pc(\mathbb R^d)$ be the family of distributions on $\mathbb R^d$ and consider the following $D=d(d+1)/2+1$ real-valued continuous functions on $\Pc(\mathbb R^d)$:
\begin{align}
    f_{k\ell}(P_{\Xm}) & = \E{X_k X_\ell}, \;\;
    1\le k\le d,\; 1\le \ell\le k \label{eq:Caratheodory-fkl} \\
    s_\lambda(P_{\Xm}) & = s_\lambda(\Xm)
\end{align}
where $P_{\Xm}\in\Pc(\mathbb R^d)$ and where we have abused notation by writing $s_\lambda(P_{\Xm})$ with argument $P_{\Xm}$ rather than $\Xm$. Then for any probability distribution $P_{\Vm}$ on the Borel $\sigma$-algebra of $\Pc(\mathbb R^d)$ there are $D$ distributions $P_{\Xm(i)}$, $i=1,\dots,D$, in $\Pc(\mathbb R^d)$ and a random variable $\Vm'$ with alphabet $\{1,\dots,D\}$ such that
\begin{align}
    \E{X_k X_\ell} & = \int_{\Pc(\mathbb R^d)} f_{k\ell}(P) \, P_{\Vm}(dP) \nonumber \\
    & = \sum_{i=1}^D P_{\Vm'}(i)\, f_{k\ell}(P_{\Xm(i)})
    \label{eq:Caratheodory1}
\end{align}
for $1\le k\le d$, $1\le \ell\le k$ and
\begin{align}
    s_\lambda(\Xm|\Vm) & = \int_{\Pc(\mathbb R^d)} s_\lambda(P) \, P_{\Vm}(dP) \nonumber \\
    & = \sum_{i=1}^D P_{\Vm'}(i)\, s_\lambda(P_{\Xm(i)}).
    \label{eq:Caratheodory2}
\end{align}
The right-hand side of \eqref{eq:Caratheodory2} is $s_\lambda(\Xm'|\Vm')$ where the distribution of $\Xm'$ conditioned on the event $\{\Vm'=i\}$ is $P_{\Xm(i)}$.
\end{proof}

We now re-prove a key result from~\cite{Liu-Viswanath-IT07}, which states that a Gaussian $\Xm$ is optimal for the problem \eqref{eq:VlQ} and one does not require $\Vm$. This theorem was also re-proved in~\cite[Thm.~1]{Geng-Nair-IT14} through a series of propositions. Our proof follows similar steps, but we do not require the existence of a maximizer, and we replace the independence result~\cite[Prop.~2]{Geng-Nair-IT14} with Theorems~\ref{thm:unboundedstability-2} and~\ref{thm:entropies-2}.

\begin{theorem}[See {\cite[Thm.~8]{Liu-Viswanath-IT07}}] \label{theorem:LV-IT07-Thm8}
 If $\lambda > 1$ then we have
 $V_\lambda(\Qm)=s_\lambda(\Xm_g)$ for some $\Xm_g \sim \Nc({\bf 0},\widehat \Qm)$ with $\widehat \Qm \preceq \Qm$.
\end{theorem}

\begin{proof}
We may again assume $\E{\Xm}={\bf 0}$. Consider a product AGN broadcast channel with sub-channels $i=1,2$ for which the channel outputs are
\begin{align}
    \Ym_{1i} & = \Gm_1 \Xm_i + \Zm_{1i} \label{eq:product-bc-1} \\
    \Ym_{2i} & = \Gm_2 \Xm_i + \Zm_{2i} \label{eq:product-bc-2}
\end{align}
where $\Xm_1,\Xm_2 \sim \Nc({\bf 0},\Qm_{\Xm})$ and $\Zm_{11}, \Zm_{12}, \Zm_{21}, \Zm_{22} \sim \Nc({\bf 0},\Id_d)$ are mutually independent.
Define (see~\eqref{eq:015} and~\cite[p.~2091]{Geng-Nair-IT14})
\begin{align}
    s_\lambda(\Xm|\Vm) & = V_\lambda\left( \Qm \right) - \epsilon \label{eq:sl-close} \\
    (\Vm_i,\Xm_i) & \sim P_{\Vm,\Xm}, \quad i=1,2 \\
    \Vm_{12} & = (\Vm_1,\Vm_2)
\end{align}
where $(\Vm_1,\Xm_1)$ and $(\Vm_2,\Xm_2)$ are independent. Also consider the $\Xm_{+},\Xm_{-}$ in \eqref{eq:010} and define
\begin{align}
    & \Ym_{k+} :=\frac{1}{\sqrt{2}}(\Ym_{k1}+\Ym_{k2}), && \Ym_{k-}:=\frac{1}{\sqrt{2}}(\Ym_{k1}-\Ym_{k2}) \label{eq:010-02-bc} \\
    &\Zm_{k+} :=\frac{1}{\sqrt{2}}(\Zm_{k1}+\Zm_{k2}), && \Zm_{k-}:=\frac{1}{\sqrt{2}}(\Zm_{k1}-\Zm_{k2})\label{eq:010-03-bc}
\end{align}
for $k=1,2$ so that
\begin{align}
    & \Ym_{k+} = \Gm_k \Xm_{+}+\Zm_{k+},
    && \Ym_{k-} = \Gm_k \Xm_{-}+\Zm_{k-}
\end{align}
and the noise vectors $\Zm_{1+}, \Zm_{1-}, \Zm_{2+}, \Zm_{2-} \sim \Nc({\bf 0},\Id_d)$ are mutually independent. As a final definition, consider the expression
\begin{align}
    & s_\lambda(\Xm_\ell,\Xm_m|\Vm) \nonumber \\
    & := I(\Xm_\ell, \Xm_m;\Ym_{1\ell}, \Ym_{1m} | \Vm) - \lambda I(\Xm_\ell, \Xm_m;\Ym_{2\ell}, \Ym_{2m} | \Vm)
    \label{eq:slm}
\end{align}
for $(\ell,m)=(1,2)$ and $(\ell,m)=(+,-)$. We study cases where given $\Vm$ we have the Markov chain
\begin{align}
    (\Ym_{1\ell},\Ym_{2\ell}) - \Xm_\ell - \Xm_m - (\Ym_{1m},\Ym_{2m}).
\end{align}
We can thus expand \eqref{eq:slm} as (see~\cite[p.~2090]{Geng-Nair-IT14})
\begin{align}
    & I(\Xm_\ell ;\Ym_{1\ell} | \Vm, \Ym_{2m}) + I(\Xm_m;\Ym_{1m} | \Vm, \Ym_{1\ell}) \nonumber \\
    & - \lambda I(\Xm_\ell;\Ym_{2\ell} | \Vm, \Ym_{2m}) - \lambda I(\Xm_m; \Ym_{2m} | \Vm, \Ym_{1\ell}) \nonumber \\
    & \quad - (\lambda-1) I(\Ym_{1\ell} ; \Ym_{2m} | \Vm).
    \label{eq:slm-2}
\end{align}
With the above definitions, we have
\begin{align}
& 2V_\lambda(\Qm) = s_\lambda(\Xm_1|\Vm_1)+s_\lambda(\Xm_2|\Vm_2)+2\epsilon \nonumber \\
& = s_\lambda(\Xm_1,\Xm_2 | \Vm_{12})+2\epsilon \nonumber \\
& = s_\lambda(\Xm_{+},\Xm_{-} | \Vm_{12})+2\epsilon \nonumber \\
& \overset{(a)}{=} s_\lambda(\Xm_{+}|\Vm_{12},\Ym_{2-})+s_\lambda(\Xm_{-}|\Vm_{12},\Ym_{1+}) \nonumber \\
& \qquad -(\lambda-1) I(\Ym_{1+};\Ym_{2-} | \Vm_{12})+2\epsilon \nonumber \\
& \le \underbrace{S_\lambda(\Xm_{+})}_{\displaystyle \le V_\lambda(\Qm)}+\underbrace{S_\lambda(\Xm_{-})}_{\displaystyle \le V_\lambda(\Qm)}-(\lambda-1)I(\Ym_{1+};\Ym_{2-} | \Vm_{12})+2\epsilon \nonumber \\
& \overset{(b)}{\le} 2V_\lambda(\Qm) - (\lambda-1) I(\Ym_{1+};\Ym_{2-} | \Vm_{12})+2\epsilon
\label{eq:bc:001}
\end{align}
where step $(a)$ follows from \eqref{eq:slm-2} and step $(b)$ follows by $\E{\Xm_{+}\Xm_{+}^T}\preceq\Qm$ and $\E{\Xm_{-}\Xm_{-}^T}\preceq\Qm$. We thus have $I(\Ym_{1+};\Ym_{2-} | \Vm_{12})\le 2\epsilon$.

Lemma~\ref{lemma:Caratheodory-bound} states that $\Vm$ can have a finite alphabet with $d(d+1)/2+1$ letters. We are now faced with the problem that $I(\Ym_{1+};\Ym_{2-} | \Vm_{12})$ is the expectation of $I(\vv_{12}):=I(\Ym_{1+};\Ym_{2-} | \Vm_{12}=\vv_{12})$ with respect to $P(\vv_{12})$, and some $I(\vv_{12})$ could be large. We wish to bound the probability of large $I(\vv_{12})$, so define the event $\set{E}:=\{I(\Vm_{12}) \ge \gamma\cdot\epsilon \}$. The Markov inequality gives
\begin{align}
    \Pr{\set{E}}
    \le \frac{\E{I(\Vm_{12})}}{\gamma \cdot \epsilon}
    \le \frac{2}{\gamma(\lambda-1)}
    \label{eq:PE-bound}
\end{align}
where the second step follows by \eqref{eq:bc:001}. We can choose $\gamma,\epsilon$ so that $\gamma$ is large, $\epsilon$ is very small and $\gamma\cdot\epsilon$ is small, e.g., we choose $\gamma=1/\sqrt{\epsilon}$.

Next, if $I(\vv_{12})$ is small then $\Xm_1$ and $\Xm_2$ are approximately Gaussian conditioned on $\Vm_{12}=\vv_{12}$. However, we wish to find a significant subset of such $\vv_{12}$ for which the covariance matrices of $\Xm_1,\Xm_2$ are the same, and for this we require the subset to have the form $\set{S}\times\set{S}$. Consider the following set of high-probability letters:
\begin{align}
    \set{S} = \left\{\vv:P(\vv) > \sqrt{\frac{2}{\gamma(\lambda-1)}} \right\}
\end{align}
where $\gamma$ is sufficiently large so that $\set{S}$ has at least one letter. The pairs $\vv_{12}=(\vv_1,\vv_2)$ in $\set{S}\times\set{S}$ thus have high probability:
\begin{align}
    \Pr{\Vm_{12} \in \set{S}\times\set{S}}
    & = \Pr{\Vm\in\set{S}}^2 \nonumber \\
    & \ge 1 - 2 \sum_{\vv\in\set{S}^c} P(\vv) \nonumber \\
    & \overset{(a)}{\ge} 1 - d(d+1) \sqrt{\frac{2}{\gamma(\lambda-1)}}
    \label{eq:SxS}
\end{align}
where step $(a)$ follows because there are at most $d(d+1)/2$ letters in $\set{S}^c$. Moreover, for all $\vv_{12}\in\set{S}\times\set{S}$ we have
\begin{align}
    P(\vv_{12}) > \frac{2}{\gamma(\lambda-1)}
    \overset{(a)}{\ge} \Pr{\set E}
     = \sum_{\vv_{12}':I(\vv_{12}') \ge \gamma\cdot\epsilon} P(\vv_{12}')\label{eq:SxS-01}
\end{align}
where $(a)$ follows by \eqref{eq:PE-bound}. We thus have $I(\vv_{12})<\gamma\cdot\epsilon$ so that $I(\vv_{12})$ is small, as desired.

So suppose $\vv_{12}\in\set{S}\times\set{S}$. By Theorem~\ref{thm:unboundedstability-2} the c.f.s of $\Xm_1$ (conditioned on $\Vm_1=\vv_1$) and $\Xm_2$ 
(conditioned on $\Vm_2=\vv_2$) are close to two Gaussian c.f.s with the same covariance matrix $\widehat \Qm$. Moreover, since $\vv_1$ and $\vv_2$ are arbitrary in $\set{S}$, we can choose the same $\widehat \Qm$ for all $\vv_{12}\in\set{S}\times\set{S}$, but in general we must decrease $\sqrt{\epsilon}$ for the same approximation precision.\footnote{Formally, this follows in two steps. First, fix $\epsilon$ and apply Theorem~\ref{thm:boundedstability} with the same $\widehat \Qm$ for a fixed $\vv_1\in\set{S}$ and for any $\vv_2\in\set{S}$. In a second step, apply Theorem~\ref{thm:boundedstability} with the same $\widehat \Qm$ for each fixed $\vv_2\in\set{S}$ and any $\vv_1\in\set{S}$. This second step requires reducing $\epsilon$ for the same approximation precision, for the same reason that the approximation \eqref{eq:bern:026a} is weaker than~\eqref{eq:bern:026}.} We abuse notation by continuing to use the same $\sqrt{\epsilon}$.
The bound \eqref{eq:bern:052-03-entropy-2} of Theorem~\ref{thm:entropies-2} gives
\begin{align}
    h(\Ym_{11}|\Vm_{12}=\vv_{12}) & \le 
    \frac{1}{2} \log \det \left( \Gm_1 \widehat \Qm \Gm_1^T + \Id_d \right) \nonumber \\
    & \quad + \frac{d}{2} \log(2\pi e) + B(\epsilon') 
    \label{eq:hY11-bound} \\
    h(\Ym_{21}|\Vm_{12}=\vv_{12}) & \ge 
    \frac{1}{2} \log \det \left( \Gm_2 \widehat \Qm \Gm_2^T + \Id_d \right) \nonumber \\
    & \quad + \frac{d}{2} \log(2\pi e) - B(\epsilon')
    \label{eq:hY21-bound} 
\end{align}
where we now must choose $\epsilon'=2\epsilon^{1/4}$. 

By~\eqref{eq:bern:052-03-entropy-Q-Y} with $\E{\Xm}={\bf 0}$ we have
\begin{align}
    & \Gm_k \widehat \Qm \Gm_k^T + \Id_d + \widehat \mv_k \widehat \mv_k^T \nonumber \\
    & \preceq \E{\Ym_{k1} \Ym_{k1}^T | \Vm_{12}=\vv_{12}} + B(\epsilon') \Id_d \nonumber \\
    & = \Gm_k \E{\Xm_1 \Xm_1^T | \Vm_{12}=\vv_{12}} \Gm_k^T + \left(1+B(\epsilon')\right) \, \Id_d
    \label{eq:bc-Y-Q-bound}
\end{align}
for $k=1,2$. The bounds \eqref{eq:hY11-bound}-\eqref{eq:hY21-bound} give
\begin{align}
    & s_\lambda(\Xm_1|\Vm_{12}=\vv_{12})
     \nonumber \\
    & = h(\Ym_{11}|\Vm_{12}=\vv_{12}) - \lambda h(\Ym_{21}|\Vm_{12}=\vv_{12}) \nonumber \\
    & \quad - (d/2) (1-\lambda) \log(2 \pi e) \nonumber \\
    & \le s_\lambda(\Xm_g) + B(\epsilon')(1+\lambda)
    \label{eq:s-lambda-v12-bound}
\end{align}
where $\Xm_g \sim \Nc\Big(\widehat \mv,\widehat \Qm\Big)$ and thus
\begin{align}
    s_\lambda(\Xm_g)
    & := \frac{1}{2} \log \det \left( \Gm_1 \widehat \Qm \Gm_1^T + \Id_d \right) \nonumber \\
    & \quad -\frac{\lambda}{2} \log \det \left( \Gm_2 \widehat \Qm \Gm_2^T + \Id_d \right).
\end{align}
Moreover, we can re-write \eqref{eq:bc-Y-Q-bound} as
\begin{align}
    \widehat \Qm
    \preceq \E{\Xm_1 \Xm_1^T | \Vm_{12}=\vv_{12}} + B(\epsilon') \left(\Gm_k^{T} \Gm_k\right)^{-1}.
    \label{eq:bc-Y-Q-bound-2}
\end{align}
Taking an expectation w.r.t. $\Vm_{12}=\vv_{12}$ and applying the constraint $\E{\Xm_1 \Xm_1^T} \preceq \Qm$ we have 
\begin{align}
    \widehat \Qm
    & \preceq \E{\Xm_1 \Xm_1^T}
    + B(\epsilon') \left(\Gm_k^{T} \Gm_k\right)^{-1} \nonumber \\
    & \preceq \Qm
    + B(\epsilon') \left(\Gm_k^{T} \Gm_k\right)^{-1}.
    \label{eq:widehat-Qm-bound}
\end{align}

Combining the above results, we have
\begin{align}
    & s_\lambda(\Xm_1|\Vm_1) 
    \overset{(a)}{=} s_\lambda(\Xm_1|\Vm_{12}) \nonumber \\
    & \overset{(b)}{\le} \sum_{\vv_{12}\in\set{S}\times\set{S}} P(\vv_{12}) s_\lambda(\Xm_1|\Vm_{12}=\vv_{12}) \nonumber \\
    & \qquad + \sum_{\vv_{12}\in(\set{S}\times\set{S})^c} P(\vv_{12}) \cdot \frac{1}{2} \log \frac{\det\left(\Qm_1' + \Qm_2' \right)}{\det \Qm_1'} \nonumber \\
    & \overset{(c)}{\le} 
    s_\lambda(\Xm_g)
    + B(\epsilon')(1+\lambda) \nonumber \\
    & \qquad +  d(d+1) \sqrt{\frac{2}{\gamma(\lambda-1)}} \cdot \frac{1}{2} \log \frac{\det\left(\Qm_1' + \Qm_2' \right)}{\det \Qm_1'}
    \label{eq:s-lambda-bound-2}
\end{align}
where step $(a)$ follows because $(\Xm_1,\Vm_1)$ and $\Vm_2$ are independent; step $(b)$ follows by Lemma~\ref{lemma:s_lambda-bound}; and step $(c)$ follows by~\eqref{eq:SxS} and~\eqref{eq:s-lambda-v12-bound}. Since \eqref{eq:widehat-Qm-bound}-\eqref{eq:s-lambda-bound-2} are valid for any $\epsilon>0$, we find that
\begin{align}
    V_\lambda(\Qm)
    \le \max_{\widehat \Qm \preceq \Qm} \; & s_\lambda(\Xm_g)
    \label{eq:V-lambda-bound}
\end{align}
which is the desired result.
\end{proof}

Theorem~\ref{theorem:LV-IT07-Thm8} can be used to show that Gaussian signaling is optimal for two-receiver broadcast channels with dedicated (also called private) messages for each receiver; see~\cite[Sec.~IV.A]{Liu-Viswanath-IT07} and~\cite[Sec.~III.A]{Geng-Nair-IT14}. Moreover, the method described above extends to two-receiver broadcast channels with a common message since the proof of Theorem 2 in~\cite{Geng-Nair-IT14} again relies on bounding a term of the form $I(\Ym_{1+};\Ym_{2-} | \Vm_{12})$.

We remark that~\cite[Sec.~II.B]{Geng-Nair-IT14} and Theorem~\ref{theorem:LV-IT07-Thm8} treat the case $\lambda>1$ while~\cite[Thm.~8]{Liu-Viswanath-IT07} includes $\lambda=1$. However, as pointed out in~\cite[Remark~9]{Geng-Nair-IT14}, the case $\lambda=1$ can be treated by showing that a capacity function is convex and bounded while the case $\lambda<1$ can be treated by reversing the roles of $\Ym_1$ and $\Ym_2$.

\section{Conclusions}
\label{sec:conclusions}

The stability of Bernstein's characterization of Gaussian distributions was extended to vectors. Refined stability results were derived for vectors with AGN. The theory was used to develop a soft doubling argument that establishes the optimality of Gaussian vectors for point-to-point and product channels with AGN, and for a classic extremal inequality.

We conclude with a few remarks.
\begin{itemize}
    \item It seems that soft doubling can replace hard doubling in general. However, whether soft doubling can provide new inequalities and capacity theorems that hard doubling cannot remains to be seen. For example, if a communications model has a strict cost constraint such as $\E{\|\Xm\|^2}<P$, then one can turn to stability (as in Theorems~\ref{thm:unboundedstability},~\ref{thm:entropies},~\ref{thm:unboundedstability-2},~\ref{thm:entropies-2}) rather than, e.g., relaxing the constraint to $\E{\|\Xm\|^2}\le P$, proving the existence of a maximizer (if possible) and applying Theorem~\ref{thm:Bernstein}.  In this sense, stability is more flexible than requiring the existence of maximizing distributions, just as suprema are more flexible than maxima.
    \item Soft doubling provides capacity bounds for non-Gaussian distributions, such as those for finite modulation alphabets. For example, a non-Gaussian $\Xm$ will require $\epsilon$ in~\eqref{eq:018} to be lower bounded by a positive number. However, the constants $C(\epsilon)$ and $\tilde C$ in Theorems~\ref{thm:boundedstability} and~\ref{thm:unboundedstability} will be large in general, so the new capacity bound will hardly improve the bound with Gaussian $\Xm$.
    \item One disadvantage of soft doubling is that one must work with inequalities and perturbations, which leads to additional proof steps. For example, in Sec.~\ref{sec:point2point} the bound \eqref{eq:015} with $\epsilon=0$ shows that the best $\Xm$ is Gaussian, while a few more steps are needed for $\epsilon>0$. Similarly, in Sec.~\ref{sec:AGN-Product} we needed to develop a new device to transfer $\epsilon$-dependence of the channel output vectors to the input vectors, and in Sec.~\ref{sec:broadcast} we needed to treat conditioning more carefully than for $\epsilon=0$.
\end{itemize}
The above remarks point out that hard and soft doubling each have their advantages and disadvantages, and which one to use to prove capacity theorems is a matter of preference. Finally, future work includes proving stability for generalizations of Bernstein's theorem such as in~\cite{darmois53,Skitovic53,ghurye1962characterization}.

\appendices
\section{Proof of Lemma \ref{lemma:hyperstability}}\label{appendix:a}

This appendix reviews results from \cite{hyers1941stability,skof1983proprieta,kominek1989local} on the stability of Cauchy's functional equation. For a non-negative $\theta$, the complex-valued function $g$ is called $\theta$-additive in $\set{E} \subseteq \mathbb R^d$ if
\begin{align}
    | g(\xv+\yv) - g(\xv) - g(\yv) | \le \theta
\end{align}
for all $\xv, \yv \in \set{E}$ such that $\xv+\yv \in \set{E}$. The function $g$
 is called additive in $\set E$ if it is $0$-additive in $\set E$.
 
The following Lemmas prove the existence of a linear function that is $\theta$-additive in various sets $\set{E}$. Lemma~\ref{lemma:Hyers-stability} is a classic result of Hyers for $\set E = \mathbb R$ that answered a question of Ulam. Lemmas~\ref{lemma:Skof-01} and~\ref{lemma:Skof-02} use ``tiling'' to apply Hyers' result to $\set E = \mathbb R^+ := \{x \in \mathbb R: x \geq 0\}$ and $\set E = [-T,T)$, $T>0$. Finally, Lemma~\ref{lemma:kominek-local} extends Lemma~\ref{lemma:Skof-02} to multiple dimensions and is slightly more general than Lemma \ref{lemma:hyperstability}.

\begin{lemma}[See~{\cite{hyers1941stability}}]
\label{lemma:Hyers-stability}
    If $g$ is $\theta$-additive in $\mathbb R$ then the limit
    \begin{align}
        G(x) := \lim_{n\rightarrow \infty} 2^{-n}g(2^n x)
    \end{align}
    exists for each $x \in \mathbb R$ and $G$ is the unique additive function in $\mathbb R$ such that
    \begin{align}
        |g(x) - G(x)| \leq \theta, \quad \forall \; x \in \mathbb R.
    \end{align}
    Moreover, if $g$ is continuous in at least one point, then $G$ is continuous and linear in $\mathbb R$.\footnote{From \cite[Theorem 2]{hyers1941stability}, if $g$ is continuous in at least one point, then $G$ is continuous in $\mathbb{R}$. Moreover, if $G$ is additive and continuous in $\mathbb{R}$, then $G$ is linear in $\mathbb{R}$.}
\end{lemma}

\begin{lemma}[See~{\cite{skof1983proprieta}}]
\label{lemma:Skof-01}
    If $g$ is $\theta$-additive in $\mathbb R^+$ then there is an additive function $G$ in $\mathbb R$ such that
    \begin{align}
        |g(x)-G(x)| \leq \theta, \quad \forall \; x \in \mathbb R^+.
    \end{align}
    Moreover, if $g$ is continuous in at least one point in $\mathbb R^{+}\setminus\{0\}$ then $G$ is continuous and linear in $\mathbb R$.
\end{lemma}
\begin{proof}
    Define the function $\tilde{g}:\mathbb R \rightarrow \mathbb  C$ such that $\tilde{g}(x)=g(x)$ for $x\ge 0$, and $\tilde{g}(x)=-g(-x)$ for $x <0$. If $x$ and $y$ have the same sign, then we have
    \begin{align}
        |\tilde{g}(x+y)-\tilde{g}(x)-\tilde{g}(y)| \le \theta.
        \label{eq:Skof-Lemma1-1}
    \end{align}
    It remains to check the case $x\ge0$ and $y<0$, so the left-hand side of \eqref{eq:Skof-Lemma1-1} is one of
    \begin{align}
    \begin{array}{ll}
        |g(x+y)-g(x)+g(-y)|, & x+y \ge 0 \\
        |-g(-x-y)-g(x)+g(-y)|, & x+y < 0.
    \end{array}
    \label{eq:Skof-Lemma1-2}
    \end{align}
    Now define $\tilde y=-y$, $z=x+y$, and $\tilde z=-z$ and write the two expressions in \eqref{eq:Skof-Lemma1-2} as
    \begin{align}
    \begin{array}{ll}
        |g(\tilde y+z)-g(\tilde y)-g(z)| \le \theta, & \tilde y, \, z \ge 0 \\
        |g(x+\tilde z)-g(x)-g(\tilde z)| \le \theta, & x, \, \tilde z \ge 0
    \end{array}
     \label{eq:Skof-Lemma1-3}
    \end{align}
    where the inequalities follow because $g$ is $\theta$-additive in $\mathbb R^+$. Thus, $\tilde g$ is $\theta$-additive in $\mathbb R$ and by Lemma \ref{lemma:kominek-local} there is a unique additive function $G$ with $G(x) =\lim_{n\rightarrow \infty} 2^{-n}\tilde g(2^n x)$ such that
    \begin{align}
        |\tilde g(x)-G(x)| \leq \theta, \quad \forall \; x \in \mathbb R.
    \end{align}
    Moreover, if $g$ is continuous in at least one point in $\mathbb R^{+}\setminus\{0\}$, then $\tilde g$ is continuous in at least one point. From Lemma \ref{lemma:Hyers-stability}, $G$ is continuous and linear in $\mathbb R$.
\end{proof}

\begin{lemma}[See~{\cite{skof1983proprieta}}]
\label{lemma:Skof-02}
    If $g$ is $\theta$-additive in $\left[-T, T\right)$ then there is an additive function $G$ in $\mathbb R$ such that
    \begin{align}
        |g(x)-G(x)| \leq 3\theta, \quad \forall \; x \in \left[ -T,T \right).\label{eq:lemma-Skof-001}
    \end{align}
    Moreover, if $g$ is continuous in at least one point in $(-T,T)$, then $G$ is continuous and linear in $\mathbb R$.
\end{lemma}
\begin{proof}
    Write $x \in \mathbb R$ in the form $x=k_x T+r_x$, $k_x \in \mathbb  Z$, $0\le r_x<T$. Now define $\tilde{g}:\mathbb R \rightarrow \mathbb  C$ as
    \begin{align}
        \tilde{g}(x):=-k_x g(-T)+g(r_x), \quad \forall \; x \in \mathbb R.
    \end{align}
    First, notice that if $x \in \left[0, T\right)$ then
    \begin{align}
        |\tilde{g}(x)-g(x)| = |g(r_x)-g(x)| = 0
    \end{align}
    and if $x \in \left[-T, 0\right)$ then
    \begin{align}
        |\tilde{g}(x)-g(x)| = |g(-T)+g(r_x)-g(x)| \le \theta
    \end{align}
    where the inequality follows because $x=-T+r_x$ and $g$ is $\theta$-additive. Hence we have
    \begin{align}
        |\tilde{g}(x)-g(x)| \leq \theta, \quad \forall \; x \in \left[-T, T\right).\label{eq:lemma-Skof-01}
    \end{align}
    
    Now consider $y=k_y T+r_y$, $k_y \in \mathbb  Z$, $0\leq r_y<T$, and suppose we have $0 \leq r_x+r_y <T$. We may then write
    \begin{align}
        |\tilde{g}(x+y)-\tilde{g}(x)-\tilde{g}(y)|&=|g(r_x+r_y)-g(r_x)-g(r_y)|\nonumber\\
        &\leq \theta.\label{eq:lemma-Skof-03}
    \end{align}
    Next, if $T \leq r_x+r_y <2T$ then
    \begin{align}
        &|\tilde{g}(x+y)-\tilde{g}(x)-\tilde{g}(y)|\nonumber\\
        &=|\tilde{g}(r_x+r_y)-g(r_x)-g(r_y)|\nonumber\\
        & = |\tilde{g}(r_x+r_y)-g(r_x+r_y)+g(r_x+r_y)-g(r_x)-g(r_y)|\nonumber\\
        &\leq 2\theta
        \label{eq:lemma-Skof-04}
    \end{align}
    where the inequality follows by \eqref{eq:lemma-Skof-01} and because $g$ is $\theta$-additive.
    Thus, from \eqref{eq:lemma-Skof-03} and \eqref{eq:lemma-Skof-04}, $\tilde{g}$ is $2\theta$-additive in $\mathbb R$. Applying Lemma \ref{lemma:Hyers-stability}, there is a unique additive function $G$ with $G(x) =\lim_{n\rightarrow \infty} 2^{-n}\tilde g(2^n x)$ such that
    \begin{align}
        |\tilde{g}(x)-G(x)| \leq 2\theta, \quad \forall \; x \in \mathbb R.\label{eq:lemma-Skof-02}
    \end{align}
    Combining \eqref{eq:lemma-Skof-01} and \eqref{eq:lemma-Skof-02} gives the inequality \eqref{eq:lemma-Skof-001}. From Lemma \ref{lemma:Hyers-stability}, if $g$ is continuous in at least one point in $(-T,T)$, then $\tilde g$ is also, and hence $G$ is continuous and linear in $\mathbb R$.
\end{proof}

\begin{lemma}[See~{\cite{kominek1989local}}]
\label{lemma:kominek-local}
    If $g$ is $\theta$-additive in $\left[-T,T\right)^d$ then there is an additive function $G: \mathbb R^d \rightarrow \mathbb  C$ such that
    \begin{align}
        |g(\xv)-G(\xv)| \leq (4d-1)\theta, \quad \forall \; \xv \in \left[-T,T\right)^d.
        \label{eq:bern:024a}
    \end{align}
    Moreover, if the projections of $g$ onto each coordinate have at least one continuous point, then $G$ is continuous and linear in $\mathbb R^d$.
\end{lemma}
\begin{proof}
Define the functions $g_i$, $i=1,\dots,d$, as
\begin{align}
    g_i(x):=g(x\,\ev_i),
    \quad x \in [-T, T) 
\end{align}
and observe that these functions are $\theta$-additive in $\left[-T, T\right)$. Lemma \ref{lemma:Skof-02} ensures that there are additive functions $G_i:\mathbb R \rightarrow \mathbb  C$ such that
\begin{align}
        |g_i(x)-G_i(x)| \leq 3\theta, \quad \forall \; x \in \left[ -T,T \right)
\end{align}
for $i=1,\dots,d$. Now define $G:\mathbb R^d \rightarrow \mathbb  C$ as $G(\xv) := \sum_{i=1}^d G_i(x_i)$ and bound
\begin{align}
    |G(\xv)-g(\xv)|
    & \le \sum_{i=1}^d \left|G_i(x_i)-g_i(x_i)\right| \nonumber \\
    & \quad + \left|\left(\sum_{i=1}^d g_i(x_i)\right) - g(\xv)\right|\nonumber\\
    & \overset{(a)}{\le} 3d\theta + (d-1)\theta
\end{align}
where step $(a)$ follows by applying the following steps $d-1$ times:
\begin{align}
    &\left|\left(\sum_{i=1}^d g_i(x_i)\right) - g(\xv)\right| \nonumber\\
    & \le \left|\left(\sum_{i=1}^{d-1} g_i(x_i)\right) - g(x_1, \dots, x_{d-1}, 0)\right|\nonumber\\
    & \quad +\left|g(x_1, \dots, x_{d-1}, 0) + g_d(x_d) - g(\xv)\right|\nonumber\\
    & \overset{(b)}{\le} \left|\left(\sum_{i=1}^{d-1} g_i(x_i)\right) - g(x_1, \dots, x_{d-1}, 0)\right| + \theta \label{eq:kominek-local-01}
\end{align}
and step $(b)$ follows because $g$ is $\theta$-additive in $[-T,T)$. Considering Lemma \ref{lemma:Skof-02}, the functions $G_i$ are continuous and linear in $\mathbb R$ if the functions $g_i$ are continuous in at least one point in $[-T,T)$. In other words, if the projection of $g$ onto each coordinate has at least one continuous point, then $G$ is continuous and linear in $\mathbb R^d$.
\end{proof}

\section{Stability of Quadratic Functional Equations}\label{appendix:quadratic}

For $\theta\ge0$, the complex-valued function $g$ is called $\theta$-biadditive in $\set{E} \times \set{E}  \subseteq \mathbb R^d \times \mathbb R^d$ if
\begin{align}
    | g(\xv_1+\xv_2,\yv) - g(\xv_1,\yv)-g(\xv_2, \yv) | \le \theta \\
    | g(\xv,\yv_1+\yv_2) - g(\xv,\yv_1)-g(\xv, \yv_2) | \le \theta
\end{align}
for all $\xv_1, \xv_2, \yv \in \set{E}$ and $\xv, \yv_1, \yv_2 \in \set{E}$ such that $\xv_1+\xv_2 \in \set{E}$ and $\yv_1+\yv_2 \in \set{E}$. The function $g$ is called biadditive in $\set{E}\times\set{E}$ if it is $0$-biadditive in $\set{E}\times\set{E}$. Finally, $g$ is symmetric in $\set{E}$ if $g(\xv,\yv)=g(\yv,\xv)$ for all $\xv,\yv \in \set{E}$.

\begin{lemma}[See~{\cite{skof1987aastcsfmn},\cite[Theorem 3.3]{hyers1998stability}}]
\label{lemma:Skof-02-quadratic}
Let $g$ be $\theta$-biadditive in $\left[-T,T\right) \times \left[-T,T\right)$. Then there exists a function $G$ which is biadditive in $\left[-T,T\right) \times \left[-T,T\right)$ and such that 
\begin{align}
    |g(x, y)-G(x, y)| \le 6\theta, \quad x, y \in \left[-T,T\right).
    \label{eq:six-theta}
\end{align}
Moreover, if $g$ is symmetric in $\left[-T,T\right) \times \left[-T,T\right)$, then $G$ is symmetric in $\left[-T,T\right) \times \left[-T,T\right)$, and if $g(x,y)$ is continuous in at least one point with respect to (w.r.t.) both arguments, then $G(x,y)$ is continuous and bilinear in $\left[-T,T\right) \times \left[-T,T\right)$.
\end{lemma}
\begin{proof}
Using similar steps as in the proof of Lemma~\ref{lemma:Skof-02}, write $x \in \mathbb R$ in the form $x=k_xT+r_x$, $k_x \in \mathbb  Z$, $0\le r_x<T$. For a fixed $y \in \left[-T, T\right)$, define the function
\begin{align}
    \tilde{g}_y(x):=-k_x g_y(-T)+g_y(r_x), \quad \forall x \in \mathbb R
    \label{eq:proof-Skof-02-quadratic-01}
\end{align}
where $g_y(x):=g(x, y)$ is $\theta$-additive in $\left[-T, T\right)$. By the proof of Lemma~\ref{lemma:Skof-02} there is a unique additive function $G_y^*: \mathbb R \rightarrow \mathbb  C$ such that
\begin{align}
    |\tilde{g}_y(x)-G_y^*(x)| \le 3\theta, \quad \forall x \in \mathbb R, \; y \in [-T,T)
    \label{eq:proof-Skof-02-quadratic-02}
\end{align}
where
\begin{align}
    G_y^*(x)=\lim_{n \rightarrow \infty} 2^{-n} \tilde{g}_y(2^n x).
    \label{eq:proof-Skof-02-quadratic-02-a}
\end{align}
Also, if $g_y(x)$ is continuous in at least one point $x\in[-T,T)$ then $G_y^*(x)$ is continuous and linear in $x\in\mathbb R$.
    
Next, we prove that $G_y^*(x)$ is $\theta$-additive in $\left[-T, T\right)$ w.r.t $y$. Defining $G^*(x, y):=G_y^*(x)$ and $x_n=2^n x$ we have
\begin{align}
    &|G^*(x, y_1+y_2)-G^*(x, y_1)-G^*(x, y_2)| \nonumber\\
    &= \left|\lim_{n \rightarrow \infty} 2^{-n} \left(\tilde{g}_{y_1+y_2}(2^n x)-\tilde{g}_{y_1}(2^n x)-\tilde{g}_{y_2}(2^n x)\right)\right| \nonumber\\
    & \le \left|\lim_{n \rightarrow \infty} 2^{-n}k_{x_n} \left( g_{y_1+y_2}(-T)-g_{y_1}(-T)-g_{y_2}(-T) \right)\right|\nonumber\\
    & \quad + \underbrace{\left|\lim_{n \rightarrow \infty} 2^{-n}\left( g_{y_1+y_2}(r_{x_n})-g_{y_1}(r_{x_n})-g_{y_2}(r_{x_n}) \right)\right|}_{=0}\nonumber\\
    & \overset{(a)}{\le} \left|\frac{x}{T}\left(g_{y_1+y_2}(-T)-g_{y_1}(-T)-g_{y_2}(-T)\right)\right|\nonumber\\
    & \quad + \underbrace{\left|\lim_{n \rightarrow \infty} \frac{2^{-n}r_{x_n}}{T} \left(g_{y_1+y_2}(-T)-g_{y_1}(-T)-g_{y_2}(-T)\right)\right|}_{=0}
\end{align}
where step $(a)$ follows because $2^{-n} k_{x_n} = x/T-2^{-n} r_{x_n}/T$. Thus, since $g_y(x)$ is additive w.r.t.\ $y$ in $[-T,T)$ we have
\begin{align}
    |G^*(x, y_1+y_2)-G^*(x, y_1)-G^*(x, y_2)| \le \theta
\end{align}
for all $x, y_1, y_2 \in \left[-T, T\right)$ such that $y_1+y_2 \in \left[-T, T\right)$.

Define the function $\widetilde{G}(x, y):=k_y G^*(x,-T)+G^*(x, r_y)$ where $y \in \mathbb R$ is written as $y=k_y T+ r_y$, $k_y \in \mathbb  Z$, $0\le r_y < T$. By the proof of Lemma \ref{lemma:Skof-02}, for each $x\in[-T,T)$ there is a unique additive function $G(x, y) = \lim_{n \rightarrow \infty} 2^{-n} \widetilde{G}(x, 2^n y)$ such that 
\begin{align}
    |G^*(x, y)-G(x, y)| \le 3\theta, \quad x \in \left[-T, T\right), \; y\in\mathbb R.
    \label{eq:proof-Skof-02-quadratic-03}
\end{align}
Combining \eqref{eq:proof-Skof-02-quadratic-02} and \eqref{eq:proof-Skof-02-quadratic-03} we have \eqref{eq:six-theta}. 

To address symmetry, define $g^{\prime}_y(x):=G(x, y)-G(y, x)$.  If $g$ is symmetric in $\left[-T,T\right) \times \left[-T,T\right)$ then by \eqref{eq:six-theta} we have
\begin{align}
    |g^{\prime}_y(x)|
    & = \left| \big[G(x, y)-g(x,y) \big] - \big[G(y, x)-g(y,x) \big] \right| \nonumber \\
    & \le \left| g(x, y) - G(x,y) \right| + \left| g(y, x) - G(y,x) \right| \nonumber \\
    & \le 12\theta, \quad x\in\left[-T, T\right).
\end{align}
Since the function $g^{\prime}_y(x)$ is additive and bounded in $\left[-T, T\right)$, it is linear in $\left[-T, T\right)$, i.e., we can write $g^{\prime}_y(x):=a(y)\cdot x$ for some function $a(y)$; see~\cite{Reem2017,Darboux1880}. But since $g^{\prime}_y(y)=0$ for all $y \in \left[-T, T\right)$, we have $a(y)=0$ and $G(x, y)$ is symmetric in $\left[-T,T\right) \times \left[-T,T\right)$.
    
As stated above, if $g(x, y)$ is continuous in at least one point w.r.t. $x$, then $G^*(x, y)$ (and hence $\widetilde G(x, y)$) is continuous and linear in $x\in\mathbb R$. Now let $\delta,\epsilon>0$, consider a $\Delta x$ with $|\Delta x|<\delta$, and define $y_n=2^n y$. Using the definition of $\widetilde G(x,y)$, for $x,x+\Delta x \in [-T,T)$ we have
\begin{align}
    &|G(x+\Delta x,y)-G(x,y)| \nonumber\\
    & = \left|\lim_{n \rightarrow \infty} 2^{-n} \left(\widetilde{G}(x+\Delta x,y_n) - \widetilde{G}(x,y_n)\right)\right| \nonumber\\
    & = \bigg|\lim_{n \rightarrow \infty} 2^{-n} \Big[ k_{y_n} \underbrace{\left( G^*(x+\Delta x,-T) - G^*(x, -T)\right)}_{\displaystyle \le \epsilon \text{ by continuity of $G^*(x,y)$ in $x$}}  \nonumber \\
    & \quad \quad + G^*(x+\Delta x,r_{y_n}) - G^*(x, r_{y_n}) \Big] \bigg| \nonumber\\
    & \overset{(a)}{\le} \left| \lim_{n \rightarrow \infty}
    2^{-n} k_{y_n} \epsilon \right|
    \overset{(b)}{\le} \epsilon
\end{align}
where step $(a)$ follows by the continuity of $G^*(x,y)$ in $x$ and step $(b)$ follows because $2^{-n} k_{y_n} = y/T-2^{-n} r_{y_n}/T$. This establishes the continuity of $G(x,y)$ for $x\in[-T,T)$. Next, from \eqref{eq:proof-Skof-02-quadratic-02-a}, if $g(x, y)$ is continuous in at least one point w.r.t. $y$ then $G^*(x, y)$ is continuous in at least one point w.r.t. $y$. Thus, by Lemma~\ref{lemma:Skof-02} $G(x, y)$ is continuous and linear w.r.t $y$. Summarizing,  the continuity of $g(x, y)$ in at least one point w.r.t. both arguments ensures the continuity and bilinearity of $G(x, y)$.
\end{proof}

\begin{lemma}
\label{lemma:multivariate-biadditve}
Suppose $g$ is symmetric and $\theta$-biadditive in $\left[-T, T\right)^d \times \left[-T, T\right)^d$. Then there is a symmetric, biadditive function $G$ in $\left[-T, T\right)^d \times \left[-T, T\right)^d$ such that
\begin{align}
    |g(\xv, \yv)-G(\xv, \yv)| \le (7d^2-1)\theta
\end{align}
for all $\xv,\yv\in\left[-T, T\right)^d$.
Moreover, if the projections of $g(\xv,\yv)$ onto each coordinate are continuous in at least one point (for both arguments), then $G(\xv,\yv)$ is continuous and bilinear in $\left[-T, T\right)^d \times \left[-T, T\right)^d$.
\end{lemma}
\begin{proof}
Define the complex-valued functions
\begin{align}
    g_{ij}(x, y) & := g(x \ev_i , y \ev_j), \quad i, j=1,\dots,d.
\end{align}
Note that $g_{ij}$ is not necessarily symmetric but it is $\theta$-biadditive in $\left[-T, T\right) \times \left[-T, T\right)$. By Lemma \ref{lemma:Skof-02-quadratic}, there is a (perhaps non-symmetric) biadditive function $G_{ij}$ such that
\begin{align}
    |g_{ij}(x, y) - G_{ij}(x, y)| \le 6\theta, \quad \forall \; x, y \in \left[ -T,T \right).
    \label{eq:g-to-G}
\end{align}
Define the symmetric, biadditive function
\begin{align}
    G(\xv, \yv) & := \sum_{i,j=1}^d \frac{1}{2} \left[G_{ij}(x_i, y_j) + G_{ji}(y_j, x_i)\right]
\end{align}
and consider
\begin{align}
    & |G(\xv, \yv)-g(\xv, \yv)| \nonumber\\
    & \le \left| \sum_{i,j=1}^d \frac{1}{2} \left[G_{ij}(x_i, y_j) + G_{ji}(y_j, x_i)\right] - g_{ij}(x_i, y_j) \right| \nonumber \\
    & \quad + \left|\left(\sum_{i,j=1}^d g_{ij}(x_i, y_j) \right) - g(\xv, \yv) \right| \nonumber\\
    & \overset{(a)}{\le} \sum_{i,j=1}^d \frac{1}{2} \left|G_{ij}(x_i,y_j) - g_{ij}(x_i, y_j)\right| \nonumber \\
    & \quad + \sum_{i,j=1}^d \frac{1}{2} \left|G_{ji}(y_j,x_i) - g_{ji}(y_j, x_i)\right| \nonumber \\
    & \quad + \left|\sum_{j=1}^d \left( \sum_{i=1}^d g_{ij}(x_i, y_j) \right) - g(\xv, y_j\ev_j) \right| \nonumber\\
    &\quad+\left|\left(\sum_{j=1}^d g(\xv, y_j\ev_j)\right) - g(\xv, \yv)\right|\nonumber\\
    & \overset{(b)}{\le} 3d^2\theta + 3d^2\theta + d(d-1)\theta + (d-1)\theta
\end{align}
where step $(a)$ follows by $g_{ij}(x,y)=g_{ji}(y,x)$ 
and step $(b)$ follows from \eqref{eq:kominek-local-01} and \eqref{eq:g-to-G}. Finally, by Lemma \ref{lemma:Skof-02-quadratic}, if the projection of $g(\xv, \yv)$ onto each coordinate is continuous in at least one point (for both arguments), then $G(\xv, \yv)$ is continuous and bilinear.
\end{proof}

\section{Proof of Lemma \ref{lemma:klebanov1}}
\label{appendix:klebanov}
Take $\tv_1=\tv_2=\tv$ and $\tv_1=-\tv_2=\tv$ in \eqref{eq:bern:004} to obtain the respective
\begin{align}
    f_{1}(2\tv) & = f_{1}(\tv)^2 |f_{2}(\tv)|^2+r_{\epsilon}(\tv, \tv)
    \label{eq:005-03} \\
    f_{2}(2\tv) & = |f_{1}(\tv)|^2 f_{2}(\tv)^2+r_{\epsilon}(\tv, -\tv)
    \label{eq:005-04}
\end{align}
where 
\begin{align}
    |r_{\epsilon}(\tv, \tv)| \le \epsilon, \quad
    |r_{\epsilon}(\tv, -\tv)| \le \epsilon.
\end{align}
By taking absolute values in \eqref{eq:005-03}-\eqref{eq:005-04}, we have
\begin{align}
    |f_{1}(2\tv)-r_{\epsilon}(\tv, \tv)| =|f_{2}(2\tv)-r_{\epsilon}(\tv, -\tv)|.\label{eq:005-05}
\end{align}
and therefore
\begin{align}
    \big| |f_{1}(2\tv)| - |f_{2}(2\tv)| \big| \le 2\epsilon .
    \label{eq:005-06}
\end{align}
Next for $i=1,2$ we expand
\begin{align}
    &\left|f_{i}(\tv)^2\left(|f_{1}(\tv)|^2-|f_{2}(\tv)|^2\right)\right|\nonumber\\
    & = \underbrace{\left|f_{i}(\tv)\right|^2}_{\displaystyle \le 1} \cdot \underbrace{\big| \left| f_{1}(\tv)|-|f_{2}(\tv)\right| \big|}_{\displaystyle \le 2\epsilon} \cdot \underbrace{\left( |f_{1}(\tv)|+|f_{2}(\tv)|\right)}_{\displaystyle \le 2}
    \label{eq:005-07}
\end{align}
where we used \eqref{eq:005-06} with $2\tv$ replaced with $\tv$. The expression \eqref{eq:005-07} implies
\begin{align}
    f_1(\tv)^2 |f_1(\tv)|^2 = f_1(\tv)^2 |f_2(\tv)|^2 + r^{(1)}_{\epsilon}(\tv)
    \label{eq:005-07a} \\
    f_2(\tv)^2 |f_1(\tv)|^2 = f_2(\tv)^2 |f_2(\tv)|^2 + r^{(2)}_{\epsilon}(\tv)
    \label{eq:005-07b}
\end{align}
where $|r^{(i)}_{\epsilon}(\tv)| \leq 4\epsilon$.
By combining \eqref{eq:005-07a}-\eqref{eq:005-07b} and \eqref{eq:005-03}-\eqref{eq:005-04}, we have
\begin{align}
    f_i(2\tv) = f_i(\tv)^2 |f_i(\tv)|^2+r^{(3)}_{\epsilon, i}(\tv)
    \label{eq:005-08-01}
\end{align}
where $|r^{(3)}_{\epsilon, i}(\tv)| \leq 5\epsilon$, for $i=1,2$.

\section{Proof of Lemma~\ref{lemma:pointwise-cf-to-pdf}}
\label{appendix:c}
The lemma is clearly true if $\epsilon=0$, so suppose $\epsilon>0$. Let $\Phi_{\Zm}(\tv)=\exp(-\tv^T \Qm_{\Zm}\,\tv/2)$ where $\Qm_{\Zm}$ is invertible. We have $f_{\Ym}(\tv)=f_{\Xm}(\tv)\Phi_{\Zm}(\tv)$ and
\begin{align}
    |\Phi(\tv)| \overset{(a)}{\le} |f_{\Ym}(\tv)| + \epsilon
    \overset{(b)}{\le} |\Phi_{\Zm}(\tv)| + \epsilon
    \label{eq:appd-bound1a}
\end{align}
where step $(a)$ follows by assumption and step $(b)$ follows by $|f_{\Xm}(\tv)|\le 1$. Consider $|\Phi(\tv)|=\exp(-\tv^T \Qm\, \tv/2)$ and let $\lambda_{\rm min}$ and $\lambda_{\Zm,{\rm min}}$ be the smallest eigenvalues of $\Qm$ and $\Qm_{\Zm}$, respectively. We then have $\lambda_{\Zm,{\rm min}}>0$ and
\begin{align}
    \lambda_{\rm min} & = \min_{\|\tv\|_2=1} \frac{\tv^T \Qm \, \tv}{\|\tv\|_2^2} \nonumber \\
    & \overset{(a)}{\ge} \min_{\|\tv\|_2=1} \frac{-2\ln\left( e^{-\tv^T \Qm_{\Zm}\,\tv/2} + \epsilon \right)}{\|\tv\|_2^2} \nonumber \\
    & = -2\ln\left( e^{-\lambda_{\Zm,{\rm min}}/2} + \epsilon \right) \nonumber \\
    & := \lambda_{\Zm,\epsilon}
    \label{eq:appd-bound1b}
\end{align}
where step $(a)$ follows by \eqref{eq:appd-bound1a}. We require $\lambda_{\Zm,\epsilon}>0$
for this bound to be useful, i.e., we require
\begin{align}
    \epsilon<1-e^{-\lambda_{\Zm,{\rm min}}/2}.
\end{align}
Observe that $\lambda_{\Zm,\epsilon} \rightarrow \lambda_{\Zm,{\rm min}}$ as $\epsilon \rightarrow 0$. Using \eqref{eq:appd-bound1b} we have
\begin{align}
    |f_{\Ym}(\tv) - \Phi(\tv)| & \le |f_{\Ym}(\tv)|+|\Phi(\tv)| \nonumber \\
    & \le e^{-\lambda_{\Zm,{\rm min}}\|\tv\|_2^2/2}+e^{-\lambda_{\rm min}\|\tv\|_2^2/2} \nonumber \\
    & \le 2 e^{-\|\tv\|_2^2 \,\lambda_{\Zm,\epsilon}/2}.
    \label{eq:appd-bound1c}
\end{align}
For any $T_1 \ge 0$ define
\begin{align}
    \tilde T_1 = T_1 \sqrt{\lambda_{\Zm,\epsilon}}, \quad
    \tilde \tv = \tv \sqrt{\lambda_{\Zm,\epsilon}}
\end{align}
and consider
\begin{align}
    & (2\pi)^d \left| p(\yv) - \phi(\yv) \right|
    =\left| \int_{\mathbb R^d} e^{-j \tv^T \yv} ( f_{\Ym}(\tv) - \Phi(\tv) ) \,d\tv \right|
    \nonumber \\
    & \overset{(a)}{\le} \int_{\|\tv\|\le T_1} \epsilon \, d\tv + \int_{\|\tv\| > T_1} 2 e^{-\|\tv\|_2^2\lambda_{\Zm,\epsilon}/2} \, d\tv \nonumber \\
    & \overset{(b)}{\le} 2 T_1^d  \epsilon + 2 \left(\frac{2 \pi}{\lambda_{\Zm,\epsilon}} \right)^{d/2} \int_{\|\tilde \tv\|_{\infty}> \frac{\tilde T_1}{d}}  \frac{e^{-\frac{1}{2} \|\tilde \tv\|_2^2}}{(2\pi)^{d/2}}  \, d\tilde \tv
    \nonumber \\
    & = 2 T_1^d  \epsilon + 2 \left(\frac{2 \pi}{\lambda_{\Zm,\epsilon}} \right)^{d/2} \Pr{\max_{1\le i\le d} \, |Z_i| \ge \tilde T_1/d} \nonumber\\
    & = 2 T_1^d  \epsilon + 2 \left(\frac{2 \pi}{\lambda_{\Zm,\epsilon}} \right)^{d/2} \left[1-\left(1-2Q\left(\tilde T_1/d\right)\right)^d \right] \nonumber\\
    & \overset{(c)}{\le} 2 T_1^d  \epsilon + 4d \left(\frac{2 \pi}{\lambda_{\Zm,\epsilon}} \right)^{d/2}\,Q\left(\tilde T_1/d\right) \nonumber\\
    & \overset{(d)}{\le} 2 T_1^d  \epsilon + 4d \left(\frac{2 \pi}{\lambda_{\Zm,\epsilon}} \right)^{d/2}\,e^{-(\tilde T_1/d)^2/2}
    \label{eq:appd-bound2}
\end{align}
where step $(a)$ follows by the assumption and \eqref{eq:appd-bound1c}; step $(b)$ follows by \eqref{eq:normbound2} and
\eqref{eq:volume-1-d-ball}; step $(c)$ follows by Bernoulli's inequality; and step $(d)$ follows from the Chernoff bound for the $Q$-function. We may choose $\tilde{T}_1$ so that
\begin{align}
      e^{-(\tilde{T}_1/d)^2/2} = \epsilon \; \text{ or }
      \; \tilde{T}_1 = d \sqrt{-2\ln \epsilon}.
\end{align}
The pointwise bound \eqref{eq:appd-bound2} thus becomes
\begin{align}
    | p(\yv) - \phi(\yv) |
    & \le  \frac{2 (d\sqrt{-2\ln\epsilon})^d \epsilon + 4d  (2\pi)^{d/2}\epsilon}{(2\pi)^d \left(-2\ln\left( e^{-\lambda_{\Zm,{\rm min}}/2} + \epsilon \right)\right)^{d/2}} \nonumber\\
    & := B_1(\epsilon)
    \label{eq:appd-bound3}
\end{align}
and we have $B_1(\epsilon)\rightarrow0$ as $\epsilon\rightarrow0$.

\section{Proof of Lemma~\ref{lemma:pointwise-pdf-to-L1}}
\label{appendix:lemma:pointwise-pdf-to-L1}
The lemma is clearly true if $\epsilon=0$, so suppose $\epsilon>0$. We first prove~\eqref{eq:bern:052-03c}. Consider the bound
\begin{align}
    |p(\yv)-q(\yv)| \le p(\yv)+q(\yv)
    \label{eq:appd-bound4a}
\end{align}
and for any $T_2 \ge 0$ use \eqref{eq:bern:052-03b} to write
\begin{align}
    \Pr{\|\Ym_q\| > T_2}
    & = 1 - \int_{\|\yv\|\le T_2} q(\yv) \,d\yv \nonumber \\
    & \le \Pr{\|\Ym_p\| > T_2} + \int_{\|\yv\|\le T_2} B_1(\epsilon) \,d\yv .
    \label{eq:appd-bound4b}
\end{align}
We also have
\begin{align}
    \E{\|\Ym_p\|^2} & \ge \int_{\|\yv\|> T_2} p(\yv) \, \underbrace{\|\yv\|_2^2}_{\displaystyle \ge \|\yv\|^2/d} \, d\yv \nonumber \\
    & \ge \frac{T_2^2}{d} \Pr{\|\Ym_p\| > T_2} \label{eq:appd-bound4c}
\end{align}
and therefore
\begin{align}
	& \| p-q \| \nonumber \\
	& \overset{(a)}{\le} \int_{\|\yv\|\le T_2} B_1(\epsilon) \, d\yv + \int_{\|\yv\| > T_2} ( p(\yv) + q(\yv) ) \, d\yv \nonumber \\
	& \overset{(b)}{\le} 4 T_2^d B_1(\epsilon)
	+ \frac{2d\,\E{\|\Ym_p\|^2}}{T_2^2}
	\label{eq:appd-bound5}
\end{align}
where $(a)$ follows by \eqref{eq:bern:052-03b} and \eqref{eq:appd-bound4a} and $(b)$ follows by \eqref{eq:appd-bound4b}-\eqref{eq:appd-bound4c}.
We may choose $T_2$ so that $B_1(\epsilon)=T_2^{-(d+3)}$ and
\begin{align}
    \| p-q \|
    \le \frac{4}{T_2^3}+\frac{2d\,\E{\|\Ym_p\|^2}}{T_2^2}
    := B_2(\epsilon)
    \label{eq:appd-bound6}
\end{align}
and we have $B_2(\epsilon)\rightarrow0$ as $\epsilon\rightarrow0$.

We next prove~\eqref{eq:bern:052-03d}. Let $\xv\in\mathbb R^d$ and consider the expression
\begin{align}
    & \xv^T \left( \E{\Ym_q\Ym_q^T} - \E{\Ym_p\Ym_p^T} \right) \xv \nonumber \\
    & = \int_{\mathbb R^d} \big( q(\yv) - p(\yv) \big) \, |\xv^T \yv|^2 \, d\yv .
    \label{eq:appd-bound8}
\end{align}
We split the integral in \eqref{eq:appd-bound8} into two integrals over the regions $\|\yv\|\le T_2$ and $\|\yv\|> T_2$. For the first region, we have
\begin{align}
    & \int_{\|\yv\|\le T_2} \big( q(\yv) - p(\yv) \big) \, |\xv^T \yv|^2 \, d\yv \le 2 T_2^{d+2} B_1(\epsilon) \, \| \xv \|_2^2
    \label{eq:appd-bound8a}
\end{align}
where we used the Cauchy-Schwarz inequality and $\|\yv\|_2\le\|\yv\|$ to write $|\xv^T \yv|^2\le \|\xv\|_2^2\|\yv\|_2^2 \le \|\xv\|_2^2 \, T_2^2$.
We also have
\begin{align}
    & \int_{\|\yv\|> T_2} \big( q(\yv) - p(\yv) \big) \, |\xv^T \yv|^2 \, d\yv \nonumber \\
    & \le \E{|\xv^T \Ym_q|^2\cdot 1(\|\Ym_q\|>T_2)} \nonumber \\
    & \overset{(a)}{\le} \sqrt{\E{|\xv^T \Ym_q|^4}} \cdot \sqrt{\Pr{\|\Ym_q\|>T_2}} \nonumber \\
    & \overset{(b)}{\le} \|\xv\|_2^2 \sqrt{\E{\left|\frac{\xv^T}{\|\xv\|_2} \Ym_q\right|^4}} \cdot \sqrt{\frac{d\,\E{\|\Ym_p\|^2}}{T_2^2} + 2 T_2^d B_1(\epsilon)}
    \label{eq:appd-bound8b}
\end{align}
where step $(a)$ follows by the Cauchy-Schwarz inequality, and step $(b)$ follows by \eqref{eq:appd-bound4b}-\eqref{eq:appd-bound4c}. Observe that the first expectation in \eqref{eq:appd-bound8b} is the fourth moment of a projection of $\Ym_q$ onto a unit vector, which is bounded by assumption. For $B_1(\epsilon)=T_2^{-(d+3)}$ as above we thus have
\begin{align}
    \xv^T \E{\Ym_q\Ym_q^T} \xv \le \xv^T \left(\E{\Ym_p\Ym_p^T} + B_3(\epsilon) \, \Id_d \right) \xv
    \label{eq:appd-bound9}
\end{align}
where
\begin{align}
    B_3(\epsilon) := \frac{2}{T_2} + \sqrt{\E{\left|\frac{\xv^T}{\|\xv\|_2} \Ym_q\right|^4}} \cdot \sqrt{\frac{d\,\E{\|\Ym_p\|^2}}{T_2^2} + \frac{2}{T_2^3}}
    \label{eq:appd-bound9a}
\end{align}
and $B_3(\epsilon)\rightarrow 0$ as $\epsilon\rightarrow0$.

We remark that the Cauchy-Schwarz inequality is a particular case (with $r=s=2$) of H\"older's inequality which states that $\E{|XY|}\le\E{|X|^r}^{1/r} \E{|Y|^s}^{1/s}$ where $r,s\ge 1$ and $1/r+1/s=1$. If one chooses $r=1+\delta/2$ and $s=(2+\delta)/\delta$ in step $(a)$ of \eqref{eq:appd-bound8b}, then one can weaken the requirement of the existence of a fourth moment and permit the $r=2+\delta$ moments of $\Ym_q$ to be bounded for any $\delta>0$.

\section{Robust $\epsilon$-Dependence}\label{appendix:robust}
The Cauchy-Schwarz inequality gives
\begin{align}
    |f_{\Xm_1,\Xm_2}(\tv_1,\tv_2)|^2 \le  |f_{\Xm_1}(\tv_1)|^2 \cdot |f_{\Xm_2}(\tv_2)|^2
    \label{eq:abs-continuity}
\end{align}
so the joint c.f. $f_{\Xm_1,\Xm_2}$ is a.c.\ with respect to the product c.f. $f_{\Xm_1} f_{\Xm_2}$.
One can now strengthen Definition~\ref{def:eT-dependent} as follows, in analogy to how weak typicality can be strengthened to \emph{robust} typicality~\cite[Ch.~3.3]{masseyapplied1},~\cite[Ch.~2.4]{ElGamal2011}.

\begin{definition} \label{def:robust-e-dependence}
The random vectors $\Xm_1$ and $\Xm_2$ are \emph{robustly $\epsilon$-dependent} if
\begin{align}
    & |f_{\Xm_1,\Xm_2}(\tv_1, \tv_2)-f_{\Xm_1}(\tv_1)f_{\Xm_2}(\tv_2)| \nonumber \\
    & \le \epsilon \cdot |f_{\Xm_1}(\tv_1)|\cdot |f_{\Xm_2}(\tv_2)|
    \label{eq:eT-dependent-robust}
\end{align}
for all $\tv_1,\tv_2$, $i=1,2$.
\end{definition}

Observe that Definition~\ref{def:robust-e-dependence} implies
\begin{align}
    \log \frac{\left| f_{\Xm_1,\Xm_2}(\tv_1, \tv_2)  \right|}{\left| f_{\Xm_1}(\tv_1)  \right| \cdot \left| f_{\Xm_2}(\tv_2)  \right|} \le \log(1+\epsilon) \le \epsilon \log(e)
\end{align}
so a log ratio (or log difference) must be small and not only an additive difference. For example, if~\eqref{eq:eT-dependent-robust} is valid then one may replace the bound in~\eqref{eq:bern:004-r} with
\begin{align}
    |r_{\epsilon}(\tv_1, \tv_2)| \le \epsilon \cdot |f_1(\tv_1) | \cdot |f_1(\tv_2)| \cdot |f_2(\tv_1)| \cdot |f_2(-\tv_2)| .
    \label{eq:bern:004-r-app}
\end{align}
Inserting~\eqref{eq:bern:004-r-app} in~\eqref{eq:bern:006} and using $|\ln(1+z)|\le 3|z|/2$ we may rewrite~\eqref{eq:bern:007} as
\begin{align}
    |R_{\epsilon}(\tv_1,\tv_2)| \le 3\epsilon/2
    \label{eq:bern:007-app}
\end{align}
so the $p$ effectively becomes one. We can thus choose $T=\infty$, thereby avoiding to split the analysis into restricted ($\|\tv\|\le T$, Theorem~\ref{thm:boundedstability}) and unrestricted (Theorem~\ref{thm:unboundedstability}) domains. Robust $\epsilon$-dependence also gives a stronger bound than \eqref{eq:bern:052-03}, namely
\begin{align}
    |f_i(\tv)-\Phi_i(\tv)|\le C(\epsilon) \cdot |\Phi_i(\tv)|
    \label{eq:bern:052-03-app}
\end{align}
for all $\tv$ where $C(\epsilon)$ is as in \eqref{eq:Cep} with $p=1$.

Note that small mutual information does not necessarily imply robust $\epsilon$-dependence with small $\epsilon$, which is why we use the weaker $\epsilon$-dependence; see~Lemma~\ref{lemma:dt-to-mutualinfo}. However, Definition~\ref{def:robust-e-dependence} might be of independent interest. For example, the bound \eqref{eq:bern:052-03-app} gives a simple pointwise bound in the probability domain (cf. Lemma~\ref{lemma:pointwise-cf-to-pdf}):
\begin{align}
    \left| p_i(\xv) - \phi_i(\xv) \right|
    & = \left| \frac{1}{(2\pi)^d} \int_{\mathbb{R}^{d}} e^{-j\tv^T\xv}\left[ f_i(\tv) - \Phi_i(\tv) \right]  \,d\tv \right| \nonumber \\
    & \le \frac{1}{(2\pi)^d} \int_{\mathbb{R}^{d}} \left| f_i(\tv) - \Phi_i(\tv) \right| \, d\tv 
    \nonumber \\
    & \le \frac{1}{(2\pi)^d} \int_{\mathbb{R}^{d}} 
    C(\epsilon) \left| \Phi_i(\tv) \right| \,d\tv \nonumber \\
    & = C(\epsilon) \det\left( 2\pi \widehat \Qm \right)^{-1/2}
    \label{eq:pointwise-p1}
\end{align}
where $\widehat \Qm$ is the covariance matrix of the Gaussian c.f. $\Phi_i$. Also, we have the following lemma for the Gaussian product channel \eqref{eq:product-channel-1a}-\eqref{eq:product-channel-2a}.

\begin{lemma} \label{lemma:robustly-epsilon-dependent}
$\Ym_{11},\Ym_{22}$ are robustly $\epsilon$-dependent if and only if $\Xm_1,\Xm_2$ are robustly $\epsilon$-dependent.
\end{lemma}
\begin{proof}
The result follows by \eqref{eq:XY-epsilon-dependence} that we re-state here for convenience. Recall that $f_{\Gm\Xm}(\tv)=f_{\Xm}(\Gm^T\tv)$ and therefore
\begin{align}
    f_{\Ym_{ii}}(\tv_i) & = f_{\Xm_i}(\Gm_i^T\tv_i)\, e^{-\frac{1}{2}\|\tv\|^2}, \quad i=1,2 \nonumber \\
    f_{\Ym_{11},\Ym_{22}}(\tv_1,\tv_2) & = f_{\Xm_1,\Xm_2}(\Gm_1^T\tv_1,\Gm_2^T\tv_2)\, e^{-\|\tv\|^2}
\end{align}
so that
\begin{align}
    & \left| f_{\Ym_{11},\Ym_{22}}(\tv_1,\tv_2) - f_{\Ym_{11}}(\tv_1)f_{\Ym_{22}}(\tv_2) \right| \nonumber \\
    & = \left| f_{\Xm_1,\Xm_2}\left( \Gm_1^T\tv_1,\Gm_2^T\tv_2 \right) - f_{\Xm_1}\left( \Gm_1^T\tv_1 \right) 
    f_{\Xm_2}\left( \Gm_2^T\tv_2 \right) \right| \nonumber \\
    & \quad \cdot e^{-\|\tv\|^2}.
    \label{eq:robustly-epsilon-dependent}
\end{align}
The lemma now follows by the identity \eqref{eq:robustly-epsilon-dependent}.
\end{proof}

\bibliographystyle{IEEEbib}
\bibliography{Ref-DT}

\begin{thebibliography}{10}

\bibitem{Kac1939}
M.~Kac,
\newblock ``On a characterization of the normal distribution,''
\newblock {\em American J. Math.}, vol. 61, no. 3, pp. 726--728, 1939.

\bibitem{bernstein41}
S.~N. Bernshtein,
\newblock ``On a property characterizing {G}auss' law,''
\newblock {\em Trudy Leningr. Bolitekhn.}, vol. 217, no. 3, pp. 21--22, 1941.

\bibitem{darmois53}
G.~Darmois,
\newblock ``Analyse g\'en\'erale des liaisons stochastiques,''
\newblock {\em Rev. Inst. Intern. Stat.}, vol. 21, pp. 2--8, 1953.

\bibitem{Skitovic53}
V.~P. Skitovi{\v c},
\newblock ``On a property of the normal distribution,''
\newblock {\em Dokl. Akad. Nauk SSSR}, vol. 89, pp. 217--219, 1953.

\bibitem{ghurye1962characterization}
S.~G. Ghurye and I.~Olkin,
\newblock ``A characterization of the multivariate normal distribution,''
\newblock {\em Ann. Math. Stat.}, vol. 33, no. 2, pp. 533--541, 1962.

\bibitem{lieb1990gaussian}
E.~H Lieb,
\newblock ``Gaussian kernels have only Gaussian maximizers,''
\newblock {\em Invent. Math.}, vol. 102, pp. 179--208, Dec. 1990.

\bibitem{brascamp1976best}
H.~J. Brascamp and E.~H Lieb,
\newblock ``Best constants in Young's inequality, its converse, and its
  generalization to more than three functions,''
\newblock {\em Advances Mathem.}, vol. 20, no. 2, pp. 151--173, 1976.

\bibitem{carlen1991superadditivity}
E.~A. Carlen,
\newblock ``Superadditivity of Fisher's information and logarithmic Sobolev
  inequalities,''
\newblock {\em J. Functional Analysis}, vol. 101, no. 1, pp. 194--211, 1991.

\bibitem{nair2014extremal}
C.~Nair,
\newblock ``An extremal inequality related to hypercontractivity of Gaussian
  random variables,''
\newblock in {\em Inf. Theory Applic. Workshop}, 2014.

\bibitem{courtade2014extremal}
T.~A. Courtade and J.~Jiao,
\newblock ``An extremal inequality for long Markov chains,''
\newblock in {\em Allerton Conf. Commun., Control, Computing}, 2014, pp.
  763--770.

\bibitem{courtade2017strong}
T.~A. Courtade,
\newblock ``A strong entropy power inequality,''
\newblock {\em IEEE Trans. Inf. Theory}, vol. 64, no. 4, pp. 2173--2192, 2017.

\bibitem{liu2017information}
J.~Liu, T.~A. Courtade, P.~Cuff, and S.~Verd{\'u},
\newblock ``Information-theoretic perspectives on Brascamp-Lieb inequality and
  its reverse,''
\newblock {\em arXiv preprint arXiv:1702.06260}, 2017.

\bibitem{liu2018information}
J.~Liu,
\newblock {\em Information Theory from a Functional Viewpoint},
\newblock Ph.D. thesis, Princeton University, 2018.

\bibitem{liu2018forward}
J.~Liu, T.~A. Courtade, P.~W. Cuff, and S.~Verd{\'u},
\newblock ``A forward-reverse Brascamp-Lieb inequality: Entropic duality and
  Gaussian optimality,''
\newblock {\em Entropy}, vol. 20, no. 6, pp. 418, 2018.

\bibitem{anantharam2022unifying}
V.~Anantharam, V.~Jog, and C.~Nair,
\newblock ``Unifying the Brascamp-Lieb inequality and the entropy power
  inequality,''
\newblock {\em IEEE Trans. Inf. Theory}, pp. 1--1, 2022.

\bibitem{aras2022entropy}
A.~Efe and T.~A. Courtade,
\newblock ``Entropy inequalities and {G}aussian comparisons,''
\newblock {\em arXiv preprint arXiv:2206.14182}, 2022.

\bibitem{Geng-Nair-IT14}
Y.~Geng and C.~Nair,
\newblock ``The capacity region of the two-receiver {G}aussian vector broadcast
  channel With private and common messages,''
\newblock {\em IEEE Tran. Inf. Theory}, vol. 60, no. 4, pp. 2087--2104, 2014.

\bibitem{chong2016capacity}
H.-F. Chong and Y.-C. Liang,
\newblock ``On the capacity region of the parallel degraded broadcast channel
  with three receivers and three-degraded message sets,''
\newblock {\em IEEE Trans. Inf. Theory}, vol. 64, no. 7, pp. 5017--5041, 2016.

\bibitem{ramachandran2017feedback}
V.~Ramachandran and S.~R.~B. Pillai,
\newblock ``Feedback-capacity of degraded Gaussian vector {BC} using directed
  information and concave envelopes,''
\newblock in {\em Nat. Conf. Commun.}, 2017, pp. 1--6.

\bibitem{goldfeld2019mimo}
Z.~Goldfeld and H.~H. Permuter,
\newblock ``MIMO Gaussian broadcast channels with common, private, and
  confidential messages,''
\newblock {\em IEEE Trans. Inf. Theory}, vol. 65, no. 4, pp. 2525--2544, 2019.

\bibitem{lau2022uniqueness}
C.~W.~K. Lau, C.~Nair, and C.~Yao,
\newblock ``Uniqueness of local maximizers for some non-convex log-determinant
  optimization problems using information theory,''
\newblock in {\em IEEE Int. Symp. Inf. Theory}, 26 June - 01 July 2022, pp.
  432--437.

\bibitem{Sula20}
E.~Sula, M.~Gastpar, and G.~Kramer,
\newblock ``Sum-rate capacity for symmetric {G}aussian multiple access channels
  with feedback,''
\newblock {\em IEEE Trans. Inf. Theory}, vol. 66, no. 5, pp. 2860--2871, May
  2020.

\bibitem{el2022strengthened}
A.~El~Gamal, A.~Gohari, and C.~Nair,
\newblock ``A strengthened cutset upper bound on the capacity of the relay
  channel and applications,''
\newblock {\em IEEE Trans. Inf. Theory}, 2022.

\bibitem{gohari2021information}
A.~Gohari, C.~Nair, and D.~Ng,
\newblock ``An information inequality motivated by the Gaussian Z-interference
  channel,''
\newblock in {\em IEEE Int. Symp. Inf. Theory}, 2021, pp. 2744--2749.

\bibitem{costa2020structure}
M.~Costa, C.~Nair, D.~Ng, and Y.~N. Wang,
\newblock ``On the structure of certain non-convex functionals and the Gaussian
  Z-interference channel,''
\newblock in {\em IEEE Int. Symp. Inf. Theory}, 2020, pp. 1522--1527.

\bibitem{gastpar2019relaxed}
M.~Gastpar and E.~Sula,
\newblock ``Relaxed {W}yner’s common information,''
\newblock in {\em IEEE Inf. Theory Workshop}, 2019, pp. 1--5.

\bibitem{bross2020source}
Shraga~I. Bross,
\newblock ``Source coding with a causal helper,''
\newblock {\em Information}, vol. 11, no. 12, pp. 553, 2020.

\bibitem{xu2021vector}
Y.~Xu, X.~Guang, J.~Lu, and J.~Chen,
\newblock ``Vector {G}aussian successive refinement with degraded side
  information,''
\newblock {\em IEEE Trans. Inf. Theory}, vol. 67, no. 11, pp. 6963--6982, 2021.

\bibitem{wang2021secret}
Z.~Wang,
\newblock {\em Secret Key Agreement over Gaussian Two-way Wiretap Channel},
\newblock Master's thesis, Dept. Elect. and Comput. Eng., Univ. Toronto, 2021.

\bibitem{prokhorov56}
Y.~V. Prokhorov,
\newblock ``Convergence of random processes and limit theorems in probability
  theory,''
\newblock {\em Theory Prob. Appl.}, vol. 1, no. 2, pp. 157--214, 1956.

\bibitem{boos85}
D.~D. Boos,
\newblock ``A converse to {S}cheff{\'e}’s theorem,''
\newblock {\em Ann. Statist.}, vol. 13, no. 1, pp. 423--427, 1985.

\bibitem{scheffe47}
H.~Scheff\'e,
\newblock ``A useful convergence theorem for probability distributions,''
\newblock {\em Ann. Math. Stat.}, vol. 18, no. 3, pp. 434 -- 438, 1947.

\bibitem{Godavarti04}
M.~Godavarti and A.~Hero,
\newblock ``Convergence of differential entropies,''
\newblock {\em IEEE Trans. Inf. Theory}, vol. 50, no. 1, pp. 171--176, 2004.

\bibitem{hyers1941stability}
D.~H. Hyers,
\newblock ``On the stability of the linear functional equation,''
\newblock {\em Proc. Nat. Acad. Sci.}, vol. 27, no. 4, pp. 222--224, 1941.

\bibitem{skof1983proprieta}
F.~Skof,
\newblock ``Proprieta’locali e approssimazione di operatori,''
\newblock {\em Rendiconti del Seminario Matematico e Fisico di Milano}, vol.
  53, no. 1, pp. 113--129, 1983.

\bibitem{kominek1989local}
Z.~Kominek,
\newblock ``On a local stability of the {J}ensen functional equation,''
\newblock {\em Demonstr. Math.}, vol. 22, no. 2, pp. 499--508, 1989.

\bibitem{Ghourchian2017}
H.~Ghourchian, A.~Gohari, and A.~Amini,
\newblock ``Existence and continuity of differential entropy for a class of
  distributions,''
\newblock {\em IEEE Commun. Lett.}, vol. 21, no. 7, pp. 1469--1472, 2017.

\bibitem{lukacs77aap}
E.~Lukacs,
\newblock ``Stability theorems,''
\newblock {\em Advances in Appl. Probab.}, vol. 9, no. 2, pp. 336--361, 1977.

\bibitem{cramer36math}
H.~Cram\'er,
\newblock ``\"Uber eine {E}igenschaft der normalen {Verteilungsfunktion},''
\newblock {\em Math. Z.}, vol. 41, pp. 405--414, 1936.

\bibitem{sapogov1951izv}
N.~A. Sapogov,
\newblock ``The stability problem for a theorem of {C}ram\'er,''
\newblock {\em Izv. Akad. Nauk SSSR Ser. Mat.}, vol. 1, pp. 205--218, 1951,
\newblock English translation: Selected Trans. Math. Staist. Prob., vol. 1, pp.
  41--53, 1961.

\bibitem{sapogov1955lenin}
N.~A. Sapogov,
\newblock ``The problem of stability for a theorem of {C}ram\'er,''
\newblock {\em ramer. Vestnik Leningrad Univ. Mat. Meh. Astronom.}, vol. 10,
  pp. 61--64, 1955.

\bibitem{levy37}
P.~L{\'e}vy,
\newblock {\em Th{\'e}orie de l'{A}ddition des {V}ariables Al{\'e}atoires},
\newblock Gauthier-Villars, Paris, 1937.

\bibitem{sapogov1959izv}
N.~A. Sapogov,
\newblock ``On independent terms of a sum of random variables which is
  distributed almost normally,''
\newblock {\em Izv. Akad. Nauk SSSR Ser. Mat.}, vol. 5, pp. 1--31, 1959.

\bibitem{Maloshevsii68tvp}
S.~G. Maloshevskii,
\newblock ``Sharpness of an estimate of {N. A.} {S}apogov on the stability
  problem of {C}ram{\'e}r's theorem,''
\newblock {\em Teor. Verojatnost. i Primenen.}, vol. 13, pp. 522--525, 1968,
\newblock English translation: {T}heory {P}rob. {A}ppl., vol. 13, pp. 494-496,
  1968.

\bibitem{Klebanov1985}
L.~Klebanov and R.~Yanushkyavichyus,
\newblock ``$\varepsilon$-dependence of $X_1+X_2$ and $X_1-X_2$,''
\newblock {\em Litovskii Matemat. Sbornik (Lietuvos Matemat. Rinkinys)}, vol.
  25, no. 3, pp. 83--92, 1985,
\newblock English translation: {I}nst. {M}athem. {C}ybernet., {A}cademy {S}ci.
  {L}ithuanian {SSR}, pp. 236-242, 1985.

\bibitem{Klebanov1986}
L.~B. Klebanov and R.~V. Yanushkyavichyus,
\newblock ``Estimation of stability in S. N. Bernshtein’s Theorem,''
\newblock {\em Theory Probab. \& Its Appl.}, vol. 30, no. 2, pp. 383--386,
  1986.

\bibitem{Yanushkevichius98}
R.~Yanushkevichius,
\newblock ``On stability of the {B}. {G}nedenko characterization,''
\newblock {\em J. Math. Sci.}, vol. 92, no. 3, pp. 3960--3971, 1998.

\bibitem{Gabovic1976}
Y.~R. Gabovi{\v c},
\newblock ``The stability of the characterization of the normal distribution by
  the {S}kitovic-{D}armois theorem,''
\newblock {\em Zap. Naucn. Sem. Leningrad Otdel. Mat. Inst. Steklov.}, vol. 61,
  1976.

\bibitem{ushakov1999selected}
N.~G. Ushakov,
\newblock {\em Selected Topics in Characteristic Functions},
\newblock VSP, Utrecht, The Netherlands, 1999.

\bibitem{csiszar2011information}
I.~Csisz{\'a}r and J.~K{\"o}rner,
\newblock {\em Information Theory: Coding Theorems for Discrete Memoryless
  Systems},
\newblock Cambridge University Press, 2nd edition, 2011.

\bibitem{linnik1960}
Y.~V. Linnik,
\newblock ``Expansions of probability laws,''
\newblock {\em Leningrad State Univ.}, 1960.

\bibitem{sapogov1981stability}
N~.A. Sapogov,
\newblock ``Stability of solutions of functional equations connected with
  problems of characterization of probability distributions,,''
\newblock {\em J. Sov. Math.}, vol. 16, no. 2, 1981.

\bibitem{Klebanov1981}
I.~B. Klebanov and I.~A. Melamed,
\newblock ``Stability of the characterization of the normal law by properties
  of parametric estimates of distribution density,''
\newblock in {\em Problems of Stability of Stochastic Models}, pp. 60--66. VNII
  Systemic Investigations, Moscow, 1981.

\bibitem{conway1978}
J.~B. Conway,
\newblock {\em Functions of One Complex Variable I},
\newblock Graduate Texts in Mathematics. Springer, 2nd edition, 1978.

\bibitem{papoulis2002probability}
A.~Papoulis and S.~U. Pillai,
\newblock {\em Probability, Random Variables, and Stochastic Processes},
\newblock Tata McGraw-Hill Education, 2002.

\bibitem{Kim-Kramer-C22}
Y.-H. Kim and G.~Kramer,
\newblock ``Information theory for cellular wireless networks,''
\newblock in {\em Information Theoretic Perspectives on {5G} Systems and
  Beyond}, pp. 10--92. Cambridge Univ. Press, 2022.

\bibitem{Liu-Viswanath-IT07}
T.~Liu and P.~Viswanath,
\newblock ``An extremal inequality motivated by multiterminal
  information-theoretic problems,''
\newblock {\em IEEE Trans. Inf. Theory}, vol. 53, no. 5, pp. 1839--1851, May
  2007.

\bibitem{skof1987aastcsfmn}
F.~Skof and S.~Terracini,
\newblock ``On the stability of the quadratic functional equation on a
  restricted domain,''
\newblock {\em Atti Accad. Sci. Torino Cl. Sci. Fis. Mat. Natur.}, vol. 121,
  pp. 153--167, 1987.

\bibitem{hyers1998stability}
D.~H. Hyers, G.~Isac, and T.~M. Rassias,
\newblock {\em Stability of Functional Equations in Several Variables},
\newblock Springer Science \& Business Media, 1998.

\bibitem{Reem2017}
D.~Reem,
\newblock ``Remarks on the {C}auchy functional equation and variations of it,''
\newblock {\em Aequat. Math.}, vol. 91, pp. 237--264, 2017.

\bibitem{Darboux1880}
G.~Darboux,
\newblock ``Sur le th{\'e}or{\`e}me fondamental de la g{\'e}om{\'e}trie
  projective,''
\newblock {\em Math. Ann}, vol. 17, pp. 55--61, 1880.

\bibitem{masseyapplied1}
J.~L. Massey,
\newblock ``Applied Digital Information Theory {I},''
\newblock Lecture Notes 1980-1998, ETH Zurich.

\bibitem{ElGamal2011}
A.~El~Gamal and Y.-H. Kim,
\newblock {\em Network Information Theory},
\newblock Cambridge University Press, 2011.

\end{thebibliography}

\begin{IEEEbiographynophoto}{Mohammad Mahdi Mahvari}
received his B.Sc. and M.Sc. in electrical engineering from Iran University of Science and Technology and Sharif University of Technology, Tehran, Iran, in 2017 and 2019, respectively. He is currently pursuing the Dr.-Ing. degree at the Technical University of Munich, Germany. His research interests include information theory and communication theory.
\end{IEEEbiographynophoto}

\begin{IEEEbiographynophoto}{Gerhard Kramer}
 (Fellow, IEEE) received the Dr. sc. techn. degree from ETH Zurich in 1998. From 1998 to 2000, he was with Endora Tech AG, Basel, Switzerland, and from 2000 to 2008, he was with the Math Center, Bell Labs, Murray Hill, NJ, USA. He joined the University of Southern California, Los Angeles, CA, USA, as a Professor of electrical engineering in 2009. He joined the Technical University of Munich (TUM) as a Professor of communications engineering in 2010. Since 2019, he has been the TUM Senior Vice President of Research and Innovation. His research interests include information theory and communications theory, with applications to wireless, copper, and optical fiber networks. He served as the 2013 President of the IEEE Information Theory Society.
\end{IEEEbiographynophoto}

\end{document}